\documentclass{amsart}

\usepackage{xcolor}
\usepackage{graphicx}
\usepackage{amssymb}
\usepackage{latexsym}

\usepackage{mathtools}

\usepackage{tikz}

\usetikzlibrary{patterns}
\usetikzlibrary{arrows,shapes,positioning}
\usetikzlibrary{decorations.markings}
\usetikzlibrary[matrix] 

\tikzstyle directed=[postaction={decorate,decoration={markings,
    mark=at position .65 with {\arrow[arrowstyle]{latex}}}}]

\tikzstyle arrowstyle=[scale=1]

\newtheorem{theorem}{Theorem}
\newtheorem{lemma}{Lemma}
\newtheorem{proposition}{Proposition}

\theoremstyle{definition}
\newtheorem{definition}{Definition}
\newtheorem*{example}{Example} % The '*' makes it unnumbered

\newtheorem{corollary}{Corollary}
\newtheorem{remark}{Remark}

\DeclareMathOperator{\sgn}{\operatorname{sgn}}

\newcommand{\GL}{\operatorname{GL}}
\newcommand{\Wr}{\operatorname{Wr}}
\newcommand{\N}{\mathbb{N}}

\newcommand{\Z}{\mathbb{Z}}
\newcommand{\C}{\mathbb{C}}

\newcommand{\bbeta}{{\boldsymbol{\beta}}}
\newcommand{\bmu}{{\boldsymbol{\mu}}}
\newcommand{\bpi}{{\boldsymbol{\pi}}}

\newcommand{\fB}{\mathfrak{B}}
\newcommand{\fS}{\mathfrak{S}}
\newcommand{\fm}{\mathfrak{m}}
\newcommand{\fh}{\mathfrak{h}}
\newcommand{\fs}{\mathfrak{s}}
\newcommand{\fp}{\mathfrak{p}}

\newcommand{\supth}{{}^{\rm{th}}}

\newcommand{\PIV}{$\mathrm{P}_{\mathrm{IV}}\,\,$}
\newcommand{\PV}{$\mathrm{P}_{\mathrm{V}}\,\,$}
\newcommand{\cF}{\mathcal{F}}
\newcommand{\cM}{\mathcal{M}}
\newcommand{\cZ}{\mathcal{Z}}

\newcommand{\cS}{\mathcal{S}}

\newcommand{\hcZ}{\hat{\cZ}}

\newcommand{\hU}{\hat{U}}
\newcommand{\htau}{\hat{\tau}}
\newcommand{\hbmu}{{\hat{\bmu}}}
\newcommand{\ha}{\hat{a}}
\newcommand{\hh}{\hat{H}}
\newcommand{\hw}{\hat{w}}
\newcommand{\hq}{\hat{q}}
\newcommand{\hatt}{\hat{t}}

\newcommand{\MGH}{\operatorname{GH}}
\newcommand{\MO}{\operatorname{O}}
\renewcommand{\th}{\tilde{H}}

\newcommand{\h}{H}
\newcommand{\Pfour}{{\rm P_{IV}}}

\renewcommand{\boxdot}{{\ \clap{\raise0.25ex\hbox{$\bullet$}}\clap{$\square$}\ }}
\newcommand{\emptybox}{\hbox{$\square$}}
\begin{document}

\thispagestyle{empty}

\title{Rational solutions of Painlev\'e systems.}
\author{David G\'omez-Ullate}
\address{Escuela Superior de Ingenier\'ia, Universidad de C\'adiz, 11519 Puerto Real, Spain.}
 \address{Escuela Superior de Ingenier\'ia, Universidad de C\'adiz,  Avda. Universidad de C\'adiz, Campus Universitario de Puerto Real, 11519, Spain.}
\author{Yves Grandati}
\address{ Laboratoire de Physique et Chimie Th\'eoriques, Universit\'e de Lorraine, 57078 Metz, Cedex 3, France.}
\author{Robert Milson}
\address{Department of Mathematics and Statistics, Dalhousie University,
  Halifax, NS, B3H 3J5, Canada.}
\email{david.gomez-ullate@icmat.es, grandati@univ-metz.fr,   rmilson@dal.ca}

\begin{abstract}
  \noindent
  Although the solutions of Painlev\'e equations are transcendental in
  the sense that they cannot be expressed in terms of known elementary
  functions, there do exist rational solutions for specialized values
  of the equation parameters.  A very successful approach in the study
  of rational solutions to Painlev\'e equations involves the
  reformulation of these scalar equations into a symmetric system of
  coupled, Riccati-like equations known as dressing chains.  Periodic
  dressing chains are known to be equivalent to the $A_N$-Painlev\'e
  system, first described by Noumi and Yamada.  The Noumi-Yamada
  system, in turn, can be linearized as using bilinear equations and
  $\tau$-functions; the corresponding rational solutions can then be
  given as specializations of rational solutions of the KP hierarchy.

  The classification of rational solutions to Painlev\'e equations and
  systems may now be reduced to an analysis of combinatorial objects
  known as Maya diagrams.  The upshot of this analysis is a an
  explicit determinental representation for rational solutions in
  terms of classical orthogonal polynomials.  In this paper we
  illustrate this approach by describing Hermite-type rational
  solutions of Painlev\'e of the Noumi-Yamada system in terms of
  cyclic Maya diagrams.  By way of example we explicitly construct
  Hermite-type solutions for the \PIV, \PV equations and the $A_4$
  Painlev\'e system.
\end{abstract}

\maketitle

\section{Introduction}

% \begin{enumerate}
% \item Start with a standard intro regarding Painleve, rational
%   solutions. ( affine group, Backlund transformations maybe?) Cite
%   Noumi-Yamada, Tsuda, Clarkson.
% \item Basic result is classification of rational solutions of $A_{2n}$
%   chains. Painlevel IV is an application.
% \item $A_{2n+1}$ will require universal characters and the balancing
%   of the parameter.
% \item Need an intro section on Maya diagrams,  Hermite
%   pseudo-Wronskians and Laguerre pseudo-Wronskians.
% \item A section on Veselov type results regarding canonical rational
%   form and trivial monodromy. (omit?)
% \item Basic philosophy: don't do new research just gather and present
%   all results to date.
% \item One exception to the above amy be an explicit dictionary between
%   Backlund transformations and Maya diagram transformations.
% \item Present the theorem that restrictued Schur polynomials give
%   Hermite Wronskians.   This will make use of the n-variable Wronskian
%   lemma and the generating function for Hermite polynomials.
% \end{enumerate}
The defining property of the six nonlinear second order Painlev\'e
equations ${\rm P_I,\dots,P_{VI}}$ is that their solutions have fixed
monodromy; that is all movable singularities are poles.  The resulting
Painlev\'e trascendents are now considered to be the nonlinear
analogues of special functions
\cite{clarkson2003painleve,fokas2006painleve}. Although these
functions are transcendental in the sense that they cannot be
expressed in terms of known elementary functions, Painlev\'e equations
also possess special families of solutions that, for special values of
the parameters, can be expressed via known special functions such
hypergeometric functions or even rational functions
\cite{airault1979rational}.

Rational solutions of ${\rm P_{II}}$ were studied by Yablonskii
\cite{yablonskii1959rational} and Vorob'ev \cite{vorob1965rational},
in terms of a special class of polynomials that are now named after
them. For ${\rm P_{III}}$, classical solutions have been considered in
\cite{murata1995classical}.  Okamoto \cite{okamoto1987studies3}
obtained special polynomials associated with some of the rational
solutions of the fourth Painlev\'e equation (\PIV) 
\begin{equation}\label{eq:P4scalar}
  y'' = \frac{1}{2y}(y')^2+\frac{3}{2} y^3 + 4t y^2
  +2(t^2-a)y+\frac{b}{y} ,\quad y=y(t),
\end{equation}
with $a$ and $b$ constants, which are analogous to the
Yablonskii--Vorob'ev polynomials.  Noumi and Yamada
\cite{noumi1999symmetries} generalized Okamoto's results and expressed
all rational solutions of \PIV\ in terms of two types of special
polynomials, now known as the \textit{generalized Hermite polynomials}
and \textit{generalized Okamoto polynomials}, both of which maybe
given as determinants of sequences of Hermite polynomials.

A very successful approach in the study of rational solutions to
Painlev\'e equations has been the set of geometric methods developed
by the Japanese school, most notably by Noumi and Yamada,
\cite{noumi2004painleve}. The core idea is to write the scalar
equations as a set of first order coupled nonlinear system of
equations. For instance, the fourth Painlev\'e \eqref{eq:P4scalar}
equation $\Pfour$ is equivalent to the following autonomous system of
three first order equations
\begin{align}\label{eq:P4system}
f_0'&= f_0(f_1-f_2) + \alpha_0, \nonumber\\
f_1'& = f_1(f_2-f_0) + \alpha_1,\\
f_2'& = f_2(f_0-f_1) + \alpha_2, \nonumber 
\end{align}
subject to the condition
\begin{equation}\label{eq:P4normalization}
(f_0+f_1+f_2)'=\alpha_0+\alpha_1+\alpha_2=1.
\end{equation}
Once this equivalence is shown, it is clear that the symmetric form of
$\Pfour$ \eqref{eq:P4system}, sometimes referred to as $\rm{s}\Pfour$
is easier to analyze. In particular, \cite{noumi1999symmetries} showed
that the system \eqref{eq:P4system} possesses a symmetry group of
B\"acklund transformations acting on the tuple of solutions and
parameters $(f_0,f_1,f_2 |\alpha_0,\alpha_1,\alpha_2)$. This symmetry
group is the affine Weyl group $A_2^{(1)}$, generated by the operators
$\{ \mathbf{\pi},\textbf{s}_0, \textbf{s}_1, \textbf{s}_2\}$ whose
action on the tuple $(f_0,f_1,f_2 |\alpha_0,\alpha_1,\alpha_2)$ is
given by:

\begin{eqnarray}\label{eq:BT}
&&{\bf s}_k (f_j)=f_j-\frac{\alpha_k\delta_{k+1,j}}{f_k}+\frac{\alpha_k\delta_{k-1,j}}{f_k}, \nonumber\\
&&{\bf s}_k(\alpha_j)=\alpha_j-2\alpha_j\delta_{k,j}+\alpha_k(\delta_{k+1,j}+\delta_{k-1,j}),\\
&&{\bf \pi}(f_j)=f_{j+1},\qquad  {\bf \pi}(\alpha_j)=\alpha_{j+1} \nonumber
\end{eqnarray}
where $\delta_{k,j}$ is the Kronecker delta and $j,k=0,1,2 \mod(3)$.
The technique to generate rational solutions is to first identify a
number of very simple rational \textit{seed solutions}, and then
successively apply the B\"acklund transformations \eqref{eq:BT} to
generate families of rational solutions.

This is a beautiful approach which makes use of the
hidden group theoretic structure of transformations of the equations,
but the solutions built by dressing seed solutions are not very explicit,
in the sense that one needs to iterate a number of B\"acklund transformations \eqref{eq:BT} on the functions and parameters in order to obtain the desired solutions. Questions such as determining the number of zeros or poles of a given solution constructed in this manner seem very difficult to address. For this reason, alternative representations of the rational solutions have also been investigated, most notably the determinental representations \cite{kajiwara1996determinant,
  kajiwara1998determinant}.

The system of first order equations \eqref{eq:P4system} admits a
natural generalization to any number of equations, and it is known as
the $A_{N}$-Painlev\'e or the Noumi-Yamada system. The higher order
Painlev\'e system, exhibited below in \eqref{eq:Aeven} and
\eqref{eq:Aodd}, is considerably simpler (for reasons that will be
explained later), and it is the one we will focus on this paper.  The
symmetry group of this higher order system is the affine Weyl group
$A_{N}^{(1)}$, acting by B\"acklund transformations as in
\eqref{eq:BT}. The system has the Painlev\'e property, and thus can be
considered a proper higher order generalization of $\rm{s}\Pfour$
\eqref{eq:P4system}, which corresponds to $N=2$.

The next higher order system belonging to this hierarchy is the
$A_3$-Painlev\'e system, which is known to be equivalent to the scalar
\PV equation.  Rational solutions of \PV were classified using direct
analysis by Kitaev, Law and McLeod \cite{kitaev1994rational}.  These
rational solutions can be described in terms of generalized Umemura
polynomials \cite{umemura1996special,masuda2002determinant}, which
admit a description in terms of Schur functions
\cite{noumi1998umemura}.  A more general determinental representations
based on universal characters was given by Tsuda
\cite{tsuda2005universal}.  In general, these determinants are
constructed Laguerre polynomials, but for some particular values these
degenerate to Hermite polynomials and fit into the framework of the
present paper.

The $A_4$-Painlev\'e system cannot be reduced to a scalar equation,
and so represents a genuine generalization of the classic Painlev\'e
equations.  Special solutions have been studied by Clarkson and
Filipuk \cite{filipuk2008symmetric}, who provide several classes of
rational solutions via an explicit Wronskian representation, and by
Matsuda \cite{matsuda2012rational}, who uses the classical approach to
identify the set of parameters that lead to rational
solutions. However, a complete classification and explicit description
of the rational solutions of $A_{2n}$-Painlev\'e for $n\geq2$ is, to
the best of our knowledge, still not available in the literature.

Of particular interest are the special polynomials associated to these
rational solutions, whose zeros and poles structure shows extremely
regular patterns in the complex plane, and have received a
considerable amount of study
\cite{okamoto1987studies3,noumi1999symmetries,umemura1997solutions,
  fukutani2000special,clarkson2006special}. Some of these polynomial
families are known as generalized Hermite, Okamoto or Umemura
polynomials, and they can be described as Wronskian determinants of
given sequences of Hermite polynomials. We will show that all these
polynomial families are only particular cases of a larger one.

Our approach for describing rational solutions to the Noumi-Yamada
system differs from the one used by the Japanese school in that it
makes no use of symmetry groups of B\"acklund
transformations. Instead, we will be influenced by the approach of
Darboux dressing chains introduced by the Russian school
\cite{adler1994nonlinear,veselov1993dressing}, which has received
comparatively less attention in connection to Painlev\'e systems, and
which makes use of the notion of trivial monodromy
\cite{oblomkov1999monodromy}.  Our interest in rational solutions of
Painlev\'e equations follows from the recent advances in the theory of
exceptional polynomials
\cite{gomez2009extended,gomez2010extension,gomez2018durfee}, and
especially exceptional Hermite polynomials \cite{gomez2013rational}.
Nonetheless we strive to maintain the connection to the theory of
integrable systems by employing the concepts of a Maya diagram and of
bilinear relations \cite{miwa2000solitons}.

The paper is organized as follows: in Section~\ref{sec:dressing} we
introduce the equations for a dressing chain of Darboux
transformations of Schr\"odinger operators and prove that they are
equivalent to the Noumi-Yamada system. These results are well known
\cite{adler1994nonlinear} but recalling them is useful to fix notation
and make the paper self contained.  In Section~\ref{sec:Maya} we
introduce the class of Hermite-type $\tau$ functions and their
representations via Maya diagrams and pseudo-Wronskian
determinants. We introduce the key notion of \textit{cyclic Maya
  diagrams} and reformulate the problem of classifying rational
solutions of the Noumi-Yamada system as that of classifying cyclic
Maya diagrams. In Section~\ref{sec:Mayacycles} we introduce the notion
of genus and interlacing for Maya diagrams which allows us to achieve
a complete classification of $p$-cyclic Maya diagrams for any period
$p$. Finally, we apply the theory to exhibit rational solutions of the
$A_2, A_3, A_4$ systems in Section~\ref{sec:A4} and we write out
explicitly these solutions using the representation developed in the
previous sections.

% The purpose of this paper is to illustrate a new construction method
% and and explicit representation of rational solutions to higher order
% Painlev\'e systems. We conjecture that this construction includes all
% possible rational solutions to the system. A proof of this fact
% requires new arguments than the ones developed in this paper, and
% remains for now an open question.

%This paper focuses on introducing a new representation and construction method for the rational solutions of higher order Painlev\'e systems built on Wronskians of Hermite polynomials. Its purpose is to communicate in a clear and direct manner the construction for the $A_4$-Painlev\'e system.
%A longer paper in preparation \cite{longpaper} will establish the necessity part of the classification, i.e. that all rational solutions belong to this class, providing proofs for all intermediate results, and extending the classification to the $A_{2k}$-Painlev\'e system. 

% Even cyclic dressing chains provide higher order extensions of
% $\Pfive$. The situation for this even cyclic case corresponding to the
% $A_{2n+1}$-Painlev\'e systems is considerably harder. Construction
% methods similar to the ones described here are available, but the
% class of dressing chains is larger, and it includes rational
% extensions of both the harmonic and the isotonic oscillator \cite{},
% thus described by universal characters (or pairs of Maya
% diagrams). Some rational solutions have been given by Tsuda in
% \cite{tsuda2005universal} but the full classification for this case is
% still an open question.

While the Japanese school has built a beautiful framework around
Painlev\'e equations, including reductions of the KP hierarchy in Sato
theory, for the particular task of describing rational solutions of
higher order Painlev\'e equations, we find the approach of cyclic Maya
diagrams to be more direct, simple and explicit.

%%%%%%%%%%%%%%%%%%%%%%%%%%%%%%%%%%%%%
\section{ Dressing chains and Painlev\'e systems}\label{sec:dressing}
%%%%%%%%%%%%%%%%%%%%%%%%%%%%%%%%%%%%%

A factorization chain is a sequence of Schr\"odinger operators
connected by Darboux transformations.  By replacing the second-order
Schr\"odinger equations with first order Riccati equations one obtains
a closely related called a dressing chain.  The theory of dressing
chains was developed by Adler \cite{adler1994nonlinear}, Veselov and
Shabat \cite{veselov1993dressing}. The connection to Painlev\'e
equations was already noted by the just-mentioned authors, and further
developed by others \cite{takasaki2003spectral,bermudez2012complexb}.

Let us recall the well-known connection between Riccati and
Schr\"odinger equations. An elementary calculation shows that a
function $w(z)$ that satisfies a Riccati equation
\begin{equation}
  \label{eq:wUricatti}
 w' + w^2+\lambda = U 
\end{equation}
is the log-derivative of a solution $\psi(z)$ of the corresponding
Schr\"odinger equation:
\begin{equation}
  \label{eq:Lipsii}
  -\psi'' + U\psi =\lambda \psi,\qquad w =  \frac{\psi'}{\psi}.
\end{equation}

The Riccati equation \eqref{eq:wUricatti} is equivalent to the
factorization relation
\begin{equation}
  \label{eq:Lwfac}
 -D^2 + U = (D+w)(-D+w)+\lambda 
\end{equation}
It follows that a Schrodinger operator $-D^2+U$ admits a factorization
\eqref{eq:Lwfac} if and only if $w$ is the log-derivative of a formal
eigenfunction of $L$ with eigenvalue $\lambda$.

A Darboux transformation is the transformation $U\mapsto \hU$ where
\[ -D^2+\hU = (D-w)(-D-w)+\lambda \] is a second-order operator
obtained by interchanging the factors in \eqref{eq:Lwfac}.
Equivalently, the correspondence $U\mapsto \hU$ may be engendered by
the transformation $w\mapsto -w$ in \eqref{eq:Lwfac}.

Consider a doubly infinite sequence of Schr\"odinger operators
$-D^2 + U_i , i\in \Z$ where neighbouring operators are related by a
Darboux transformation
\begin{equation}
  \label{eq:Dxform}
  \begin{aligned}
    -D^2+U_i &= (D + w_i)(-D + w_i)+\lambda_i, \\
    -D^2+U_{i+1} &= (-D + w_i)(D + w_i)+\lambda_i.
  \end{aligned}
\end{equation}
Since functions $w_i$ are solutions of the Riccati equations
\begin{equation}\label{eq:Riccati}
 w_i' + w_i^2 +\lambda_i = U_i,\quad -w_i' +w_i^2+\lambda_i = U_{i+1},
 \end{equation}
 the above potentials are related by
\begin{eqnarray}
  U_{i+1} &=& U_i - 2 w'_i, \label{eq:Uplus1}\\
  U_{i+n} &=&U_i - 2 \left( w'_i+ \cdots + w'_{i+n-1}\right),\quad
  n\geq 2 \label{eq:Uplusn}
\end{eqnarray}
If we eliminate the potentials in \eqref{eq:Riccati} and set
\begin{equation}
  \label{eq:alphaidef}
  a_i =  \lambda_{i} - \lambda_{i+1}
\end{equation}
we obtain a system of coupled differential equations called the doubly
infinite dressing chain:
\begin{equation}
  \label{eq:drchain}
  (w_i + w_{i+1})' + w_{i+1}^2 - w_i^2 = a_i ,\quad i\in \Z
\end{equation}

If we impose a cyclic condition
\begin{equation}\label{eq:shift}
  U_{i+p} = U_i+\Delta,\quad i\in \Z,
  \qquad p\in \N,\; \Delta \in \C
 \end{equation}
 on the potentials of the above chain, we obtain a finite-dimensional
 system of ordinary differential equations.  If this holds, then
 necessarily $w_{i+p}=w_i$, $\alpha_{i+p}=\alpha_i$, and
\begin{equation}
  \label{eq:Deltasumalpha}
 \Delta= -(a_0 + \cdots + a_{p-1}). 
\end{equation}
Going forward, we
impose the non-degeneracy assumption that
\[ \Delta \neq 0 \] Degenerate dressing chains with $\Delta=0$ are
more closely related to elliptic functions \cite{veselov1993dressing} and
will not be considered here.

\begin{definition}
  A solution to the $p$-cyclic dressing chain with shift $\Delta$ is a
  sequence of $p$ functions $w_0,\ldots, w_{p-1}$ and complex numbers
  $a_0,\ldots, a_{p-1}$ that satisfy the following coupled
  Riccati-like equations:
\begin{equation}
  \label{eq:wachain}
  (w_i + w_{i+1})' + w_{i+1}^2 - w_i^2 = a_i ,\qquad
  i=0,1,\ldots, p-1 \mod p  
\end{equation}
subject to the condition \eqref{eq:Deltasumalpha}.  
\end{definition}

The cyclic chain has a number of evident symmetries: the
reversal symmetry
\begin{equation}
  \label{eq:reversal2}
  \hw_i = -w_{-i},\quad \ha_i = -a_{-i};
\end{equation}
the cyclic symmetry
\begin{equation}
  \label{eq:cyclicgen}
  \hw_i = w_{i + 1},\quad \ha_i= a_{i+1};
\end{equation}
and the scaling symmetry
\begin{equation}
  \label{eq:wscale}
  \hw_i(z) = k w_i(kz),\quad \ha_i = k^2
  a_i,\quad k\neq 0.
\end{equation}
In the
classification of solutions to \eqref{eq:wachain} it will be
convenient to regard solutions related by reversal, cyclic, and
scaling symmetries as being equivalent.
% Indeed,  imposing the
% normalization
% \[ \Delta = 1 \] eliminates the reversal and scaling symmetries,
% leaving only the cyclic symmetries generated by \eqref{eq:cyclicgen}.

The $p$-cyclic dressing chain is closely related to higher order
Painlev\'e systems of type $A_{N},\; N=p-1$ introduced by Noumi and
Yamada in \cite{noumi1998higher}. In the even case of $N=2n$, the
Noumi-Yamada system has the form
\begin{equation}
  \label{eq:Aeven}
  f_i' = \sum_{j=1}^{p-1} (-1)^{j+1} f_i f_{i+j} +\alpha_i,\quad
  i=0\ldots, 2n \mod 2n+1
\end{equation}
In the odd case of $N=2n-1$, the Noumi-Yamada system has a more
complicated form:
\begin{equation}
  \label{eq:Aodd}
  xf_i' = f_i\left(1    - 2\sum_{k=1}^{n-1}
    \alpha_{i+2k} + 2\sum_{j=1}^{n}\sum_{k=1}^{n-1} \sgn(2j-1-2k)
    f_{2j+i-1} f_{2k+i} \right) + 2\alpha_i \sum_{k=1}^{n-1} f_{i+2k},
\end{equation}
where $i=0,\ldots, 2n-1\mod 2n$ and $f_i = f_i(x)$. In both cases, the
parameters $\alpha_0,\ldots, \alpha_{N}$ are subject to the constraint
\[ \alpha_0+ \cdots \alpha_{N} = 1.\]
\begin{proposition}
  \label{prop:wtof}
  The $A_{2n}$ and $A_{2n-1}$ Noumi-Yamada systems \eqref{eq:Aeven}
  \eqref{eq:Aodd} are related to the $p$-cyclic dressing chain
  \eqref{eq:wachain} by the following change of variables:
  \begin{equation}
    \label{eq:ffromw}
    -\sqrt{\Delta} f_i(x) =w_{i}(z) + w_{i+1}(z),\quad \alpha_i =
    -\frac{a_i}{\Delta},\qquad z=\frac{x}{\sqrt{\Delta}},
  \end{equation}
  where $i=0,\dots, p-1 \mod p$ and where $p=2n+1$ in the first case,
  and $p=2n$ in the second case.
\end{proposition}
\begin{proof}
  The proof for the case $p=2n+1$ is quite direct. Set
  \[ d_i(x) = K(w_{i}(Kx) - w_{i+1} (Kx)), \quad K =
    \frac{1}{\sqrt{\Delta}},\] which allows us to rewrite relation
  \eqref{eq:drchain} as
  \begin{equation}
    \label{eq:fdchain}
    f_i'   = d_i f_i +\alpha_i.  
  \end{equation}
  Then,  observe that because $p$ is odd,
  \begin{equation}
    \label{eq:diodd} 
    d_i=\sum_{j=1}^{p-1} (-1)^{j+1} f_{i+j}.
  \end{equation}
  This transforms \eqref{eq:fdchain} into \eqref{eq:Aeven}.
\end{proof}

If $p=2n$ is even, the linear relation \eqref{eq:diodd} no longer
holds. Rather we have the following quadratic relation.

\begin{lemma}
  For each $i=1,\ldots, 2n \mod 2n$ we have
  \begin{equation}
    \label{eq:Aodd1}
    \begin{split}
      &      (w_{i+1}-w_i)(w_1+ \cdots + w_{2n})+ (w_i^2- w_{i+1}^2 + \cdots -
      w_{i+2n-1}^2) =\\
      &\quad =\sum_{j=1}^{n}\sum_{k=1}^{n-1} \sgn(2k+1-2j)
  (w_{2j+i-1}+ w_{2j+i}) (w_{2k+i}+ w_{2k+i+1}).
  \end{split}
\end{equation}
\end{lemma}
\begin{proof}
  The left side of \eqref{eq:Aodd1} expands to
  \begin{equation}
    \label{eq:oddproof1}
    \begin{aligned}
      \ w_{i+1}\!\! \sum_{a=2}^{2n-1} w_{i+a} - w_i
      \!\!\sum_{a=2}^{2n-1} w_{i+a}+ \sum_{a=2}^{n-1} (-1)^aw_{i+a}^2\\ 
      % =&\sum_{j=1}^{n}( w_{i+1} w_{i+2j} - w_i w_{i+2j} + w_{i+2j}^2)+
      % \sum_{j=1}^{n}( w_{i+1} w_{i+2j-1} - w_i w_{i+2j-1} -
      % w_{i+2j-1}^2)\\
      % =&\sum_{j=1}^{n-1} w_{i+1} w_{i+2j}+ \sum_{j=1}^{n-1}
      % w_{i+1}w_{i+2j+1} - \sum_{j=1}^{n-1} w_i
      % w_{i+2j}- \sum_{j=1}^{n-1} w_i 
      % w_{i+2j+1} +\sum_{j=1}^{n-1} w_{i+2j}^2-\sum_{j=1}^{n-1}w_{i+2j+1}^2
    \end{aligned}
  \end{equation}
  The right side of \eqref{eq:Aodd1} expands to
  \begin{equation}
    \label{eq:oddproof2}
    \begin{aligned}
      % & \sum_{1\leq j\le k\leq n-1}\!\!\!\!\!  f_{2j} f_{2k+1} -
      % \sum_{1\leq k< j\leq
      %   n}\!\!\!\!\! f_{2j} f_{2k+1}\\
      &\sum_{1\leq j\leq k\leq n-1}\!\!\!\!\!
      (w_{i+2j-1}+w_{i+2j})(w_{i+2k}+w_{i+2k+1}) -\sum_{1\leq k< j\leq
        n}\!\!\!\!\!(w_{i+2j-1}+w_{i+2j})(w_{i+2k}+w_{i+2k+1})\\
      =&\sum_{1\leq j< k\leq n}\!\!\!\!\!
      (w_{i+2j-1}+w_{i+2j})(w_{i+2k-2}+w_{i+2k-1}) - \sum_{1\leq k<
        j\leq
        n}\!\!\!\!\!(w_{i+2j-1}+w_{i+2j})(w_{i+2k}+w_{i+2k+1})\\
      =&\sum_{1\leq j< k\leq n}\!\!\!\!\!
      (w_{i+2j-1}+w_{i+2j})(w_{i+2k-2}+w_{i+2k-1})
      -(w_{i+2k-1}+w_{i+2k})(w_{i+2j}+w_{i+2j+1})\\
      =& \sum_{1\leq j < k \leq n} \!\!\!\!\! (w_{i+2j-1}- w_{i+2j+1})
      w_{i+2k-1} +w_{i+2j}
      (w_{i+2k-2}-w_{i+2k})+w_{i+2j-1}w_{i+2k-2} - w_{i+2j+1}w_{i+2k}\\
      =& \sum_{k=2}^n (w_{i+1}-w_{i+2k-1}) w_{i+2k-1} +
      \sum_{j=1}^{n-1} w_{i+2j}(w_{i+2j}-w_{i+2n})+
      \sum_{k=1}^{n-1}w_{i+1}w_{i+2k} -
      \sum_{j=1}^{n-1}w_{i+2j+1}w_{i+2n}\\
      =& w_{i+1}\sum_{a=2}^{2n-1} w_{i+a} - w_i \sum_{a=2}^{2n-1}
      w_{i+a}+\sum_{a=2}^{2n-1} (-1)^a w_{i+a}^2,
    \end{aligned}
  \end{equation}
  which  matches \eqref{eq:oddproof1}.
\end{proof}
\begin{proof}[Proof of Proposition \ref{prop:wtof} continued.]
Every dressing chain has an obvious first integral, obtained by
summing \eqref{eq:wachain}:
\begin{equation}
  \label{eq:wsumDelta}
  w_1(z) + \cdots + w_{2n}(z)=  - \frac12 \Delta\, z\
\end{equation}
For $p=2n$, the even-cyclic dressing chain also has an additional
first integral --- obtained by taking an alternating sum of
\eqref{eq:wachain}:
\begin{equation}
  \label{eq:Aodd2}
  2(w_1^2-w_2^2 + \cdots - w_{2n}^2) = -a_1 + a_2 - \cdots + a_{2n}.
\end{equation}
% where
% \begin{equation}
%     A_1= \alpha_1 + \alpha_3 + \cdots + \alpha_{2n-1},\qquad
%   A_2=\alpha_2 + \alpha_4 + \cdots + \alpha_{2n},\qquad
%   A_1+ A_2 = -\Delta
% \end{equation}
Using  \eqref{eq:Aodd1} and  \eqref{eq:wsumDelta} we obtain
\[ x \, d_i(x) = 2\sum_{j=1}^{n}\sum_{k=1}^{n-1} \sgn(2k+1-2j)
  f_{2j+i-1}(x) f_{2k+i}(x)+ 2\sum_{j=0}^{2n-1} (-1)^j \alpha_{i+j}
\]
Relation \eqref{eq:fdchain} may now be rewritten as
\begin{equation}
  \label{eq:Aodd3}
  x f_i'(x) = 2f_i(x) \sum_{j=1}^{n}\sum_{k=1}^{n-1} \sgn(2j-1+2k)
  f_{2j+i-1}(x) f_{2k+i}(x)- 2f_i(x)\sum_{j=0}^{2n-1} (-1)^j
  \alpha_{i+j} + \alpha_i x 
\end{equation}
Since
\[\sum_{j=1}^n f_{2j-1}(x) =   \sum_{j=1}^{n} f_{2j}(x) =\frac12 x, \]
and since
\[ \sum_{j=1}^n \alpha_{2j-1} + \sum_{j=1}^n \alpha_{2j} = 1,\]
relation \eqref{eq:Aodd3} may be rewritten as \eqref{eq:Aodd}.
\end{proof}

The problem now becomes that of finding and classifying cyclic
dressing chains, sequences of Darboux transformations that reproduce
the initial potential up to an additive shift $\Delta$ after a fixed
given number of transformations.  The theory of exceptional
polynomials is intimately related with families of Schr\"odinger
operators connected by Darboux transformations
\cite{gomez2013conjecture,garcia2016bochner}.  Each of these
exceptional operators admits a bilinear formulation in terms of
$\tau$-functions which suggests a strong connection with integrable
systems theory, and which will be the basis of the development here.
Each $\tau$-function in this class can be indexed by a finite set of
integers, or equivalently by a Maya diagram, which becomes a very
useful representation to capture a notion of equivalence and relations
of the type \eqref{eq:shift}.

%%%%%%%%%%%%%%%%%%%%%%%%%%%%%%%%%%%%%
\section{Hermite $\tau$-functions
}\label{sec:Maya}
%%%%%%%%%%%%%%%%%%%%%%%%%%%%%%%%%%%%%

In this section we introduce Hermite-type $\tau$-functions, their
bilinear relations, and 3 determinental representations of these
objects: pseudo-Wronskians, Jacobi-Trudi formula, and a Boson-Fermion
correspondence formula.  We also introduce Maya diagrams, partitions,
and indicate the relation between these two types of objects.  In a
nutshell, the Hermite-type $\tau$ function is a specialization of the
more general Schur function.  Various instances of this observations
can be found in \cite{clarkson2006specialsoliton}.  Schur functions
arise in integrable systems theory as polynomial solutions of the KP
hierarchy \cite{wilson1993bispectral}. As shown by Tsuda
\cite{tsuda2005universal}, the Painlev\'e systems are actually
reductions of the KP hierarchy.  Theorem \ref{thm:tauschur} is an
indication of how this reduction manifests at the level of solutions.
The Jacobi-Trudi formula is a determinental representation of the
classical Schur functions in terms of monomials.  A generalization to
a basis of orthogonal polynomials was introduced in
\cite{veselovsergeev}. We are interested in the specialization of this
result to the case of Hermite polynomials. A far-reaching
generalization, called the quantum Jacobi-Trudi identity was
introduced in \cite{harnadlee}.

Following Noumi \cite{noumi2004painleve}, we introduce the following.
\begin{definition}
  A Maya diagram is a set of integers $M\subset \Z$ that contains a
  finite number of positive integers, and excludes a finite number of
  negative integers.  We will use $\cM$ to denote the set of all Maya
  diagrams.

  Let $k_1 > k_2> \cdots$ be the decreasing enumeration of a Maya
  diagram $M\subset \Z$.  The condition that $M$ be a Maya diagram is
  equivalent to the condition that $k_{i+1} = k_i-1$ for $i$
  sufficiently large.  Thus, there exists a unique integer
  $\sigma_M\in \Z$ such that $k_i = -i+\sigma_M$ for all $i$
  sufficiently large. We call $\sigma_M$ to the index of $M$.
\end{definition}

We visualize a Maya diagram as a horizontally extended sequence of
$\boxdot$ and $\emptybox$ symbols with the filled symbol $\boxdot$ in
position $i$ indicating membership $i\in M$. The defining assumption
now manifests as the condition that a Maya diagram begins with an
infinite filled $\boxdot$ segment and terminates with an infinite
empty $\emptybox$ segment.

\begin{definition}
  \label{def:frobsymb}
  Let $M$ be a Maya diagram, and
  \[ M_-= \{ -m-1 \colon m\notin M, m<0\},\qquad M_+ = \{ m\colon m\in
    M\,, m\geq 0 \}. \] Let $s_1>s_2>\cdots > s_r$ and
  $t_1> t_2>\dots> t_q$ be the elements of $M_-$ and $M_+$ arranged in
  descending order.  We call the double list
  $(s_1,\ldots, s_r \mid t_q,\ldots, t_1)$ the \textit{Frobenius symbol}
  of $M$ and use $M(s_1,\ldots, s_r\mid t_q,\ldots, t_1)$ to denote the
  Maya diagram with the indicated Frobenius symbol.
\end{definition}
\noindent
It is not hard to show that $\sigma_M=q-r$ is the index of $M$.  The
classical Frobenius symbol
\cite{andrews2004integer,olsson1994combinatorics,andrews1998theory}
corresponds to the zero index case where $q=r$.

If $M$ is a Maya diagram, then for any $k\in \Z$ so is
\[ M+k = \{ m+k \colon m\in M \}.\] The behaviour of the index
$\sigma_M$ under translation of $k$ is given by
\begin{equation}\label{eq:indexshift}
M'=M+k\quad \Rightarrow \quad \sigma_{M'}=\sigma_M+k.
\end{equation}
We will refer to an equivalence class of Maya diagrams related by such
shifts as an \textit{unlabelled Maya diagram}. One can visualize the
passage from an unlabelled to a labelled Maya diagram as the choice of
placement of the origin.

\begin{definition}
  A Maya diagram $M\subset \Z$ is said to be in standard form if $p=0$
  and $t_q>0$.  Equivalently, $M$ is in standard form if the index
  $\sigma_M=q$ is the number of positive elements of $M$.
  Visually, a Maya diagram in standard form has only filled boxes
  $\boxdot$ to the left of the origin and one empty box
  $\emptybox$ just to the right of the origin. Every unlabelled Maya
  diagram permits a unique placement of the origin so as to obtain a
  Maya diagram in standard form.
\end{definition}

In \cite{gomez2018durfee} it was shown that to every Maya diagram we
can associate a polynomial called a Hermite pseudo-Wronskian.
For $n\geq 0$, let
\begin{equation}
  \label{eq:Hndef}
  H_n(x) = (-1)^n e^{x^2} \left( \frac{d}{dx}\right)^n e^{-x^2}
\end{equation}
denote the degree $n$ Hermite polynomial, and
\begin{equation}
  \label{eq:thndef}
  \th_n(x)={\rm i}^{-n} \h_{n}({\rm i}x)
\end{equation}
the conjugate Hermite polynomial.  A number of  equivalent definition of $H_n$ are available.  One is that
$y=H_n$ is the  polynomial solution of the Hermite differential
equation
\begin{equation}
  \label{eq:hermiteDE}
  y''(z)-2z y'(z) + 2n y(z)=0
\end{equation}
subject to the  normalization condition
\[ y(z) \sim 2^n z^n\quad z\to \infty.\]
Setting
\[ \hh_n(z)= e^{z^2}\th_{n}(z) \] we also note that $\hh_{-n-1}$ is a
solution of \eqref{eq:hermiteDE} for negative integers $n<0$.

A third definition involves
the 3-term recurrence relation:
\[  H_{n+1}(z) = 2z H_n(z) -2n H_{n-1}(z),\qquad H_0(z) =1,\; H_1(z)= 2z.\]
A fourth definition involves the generating function
\begin{equation}
  \label{eq:hermgf}
  \sum_{n=0}^\infty H_n(z) \frac{t^n}{n!} =  e^{2zt-t^2}.
\end{equation}

\begin{definition}
  For $s_1,\dots,s_r,t_q,\dots,t_1\in \Z$ set
  \begin{equation}
    \label{eq:taudef}
    \tau(s_1,\ldots, s_r| t_q,\ldots, t_1) = e^{-rz^2}\Wr[
    \hh_{s_1}(z),\ldots, \hh_{s_r}(z), \h_{t_q}(z),\ldots \h_{t_1}(z) ]
  \end{equation}
  where $\Wr$ denotes the Wronskian determinant of the indicated
  functions.  For a Maya diagram $M(s_1,\dots,s_r\mid t_q,\dots,t_1)$
  we let
  \begin{equation}
    \label{eq:pWdef1} \tau_M(z)  =     \tau(s_1,\ldots, s_r|
    t_q,\ldots, t_1)  \end{equation} 
  % \begin{align}
  %   \label{eq:thndef}
  %   % \th_n(z) &= {\rm i}^{-n} \h_{n}({\rm i}z) \intertext{is the
  %   %            conjugate Hermite polynomial, and where}
  % \end{align}
\end{definition}
\noindent
Note: when $r=0$ it will be convenient  to simply write
\begin{equation}
  \label{eq:tauWr}
  \tau(t_q,\ldots, t_1) = \Wr[H_{t_q},\ldots,H_{t_1}] 
\end{equation}
to indicate a Wronskian of Hermite polynomials.

The polynomial nature of $\tau_M(z)$ becomes evident once we represent
it using a slightly different determinant.

\begin{proposition} The Wronskian in \eqref{eq:taudef} admits the following
  alternative \emph{pseudo-Wronskian} representation
  \begin{equation}\label{eq:pWdef2} \tau(s_1,\ldots, s_r | t_q,\ldots,
    t_1)=
    \begin{vmatrix} \th_{s_1} & \th_{s_1+1} & \ldots &
      \th_{s_1+r+q-1}\\ \vdots & \vdots & \ddots & \vdots\\ \th_{s_r} &
      \th_{s_r+1} & \ldots & \th_{s_r+r+q-1}\\ \h_{t_q} &  \h'_{t_q} &
      \ldots & \h^{(r+q-1)}_{t_q}\\ \vdots & \vdots & \ddots & \vdots\\
      \h_{t_1} & \h'_{t_1} & \ldots & \h^{(r+q-1)}_{t_1}
    \end{vmatrix}
  \end{equation}

\end{proposition} 
\noindent
The proof of the above result can be found in \cite{gomez2018durfee}.
The term Hermite pseudo-Wronskian was also introduced in that paper,
because \eqref{eq:pWdef2} is a mix of a Casoratian and a Wronskian
determinant.  The just mentioned article also demonstrated that the
pseudo-Wronskians of two Maya diagrams related by a translation are
proportional.
\begin{proposition}\label{prop:equiv}
  Let  $\htau_M $ be the normalized pseudo-Wronskian
  \begin{equation}
    \label{eq:hHdef}
    \htau_M = \frac{(-1)^{rq}\tau(s_1,\ldots, s_r|t_q,\ldots, t_1)}{\prod_{1\leq i<j\leq r} (2s_j-2s_i)\prod_{1\leq
        i<j\leq q}
      (2 t_i-2t_j)}.
  \end{equation}
Then for any Maya diagram $M$ and $k\in\Z$ we have
  \begin{equation} \label{eq:HMequiv}
       \htau_M =  \htau_{M+k}.
  \end{equation}
\end{proposition}
\noindent
Observe that the identity in \eqref{eq:HMequiv} involves determinants
of different sizes, and a Wronskian of Hermite polynomials will not,
in general, be the smallest determinant in the equivalence class.  The
question of which determinant has the smallest size was solved in
\cite{gomez2018durfee}.

% As mentioned above, every unlabelled Maya diagram contains a Maya
% diagram in standard form, and its associated Hermite pseudo-Wronskian
% \eqref{eq:pWdef1} is just an ordinary Wronskian determinant whose
% entries are Hermite polynomials.  Thanks to Proposition
% \ref{prop:equiv}, there is no loss of generality if we consider only
% Maya diagrams in standard form, and Wronskians of Hermite polynomials.
% However, whereever possible we employ the more general
% pseudo-wronskians because of a certain conceptual clarity.

We define a partition to be a non-increasing sequence of natural
numbers $\lambda_1 \geq \lambda_2 \geq \cdots $ such that
\[ |\lambda| := \sum_{i=1}^\infty \lambda_i< \infty.\] Implicit in
this definition is the assumption that $\lambda_i=0$ for $i$
sufficiently large.  We define $\ell(\lambda)$, the length of
$\lambda$, to be the smallest $q\in \N$ such that $\lambda_{q+1}=0$.

To a partition $\lambda$ of length $q=\ell(\lambda)$ we associate the
Maya diagram $M_\lambda$ consisting of
\begin{equation}
  \label{eq:tqlambda}
  t_i = \lambda_i + q - i,\quad i=1,2,\ldots.
\end{equation}
By construction, we have
\[t_q>0,\quad \text{and}\quad t_{i+1}+1 = t_i<0,\qquad i> q. \]
Therefore $M_\lambda$ is a Maya diagram in standard form.  Indeed,
\eqref{eq:tqlambda} defines a bijection between the set of partitions
and the set of Maya diagrams in standard form.  Going forward, let
\begin{equation}
  \label{eq:taulambda}
  \tau_\lambda = \Wr[H_{t_{q}},\ldots, H_{t_1} ].
\end{equation}

For $n\in \Z$ and $\lambda$ a partition, let
\[ M^{(n)}_\lambda = M_\lambda + n -\ell(\lambda),\]
and let $t_1>t_2>\cdots$ be the decreasing enumeration of
$M^{(n)}_\lambda$. Equivalently,
\begin{equation}
  \label{eq:tnlambda}
   t_i = \lambda_i + n-i,\quad i=1,2,\ldots. 
\end{equation}
Note that the condition
$n \geq \ell(\lambda)$ holds if  and only if $M_\lambda^{(n)}$
contains all negative integers and exactly $n$ non-negative integers, that is if
\[ M_\lambda^{(n)} = M(\mid t_n,\ldots, t_1).\]

% Let $\lambda$ be a partition and $n\geq \ell(\lambda)$. We define the
% Hermite $\tau$-function to be the polynomial
% \begin{equation}
%   \label{eq:taulambda}
%   \tau^{(n)}_\lambda =  \tau(k_n,\ldots, k_1), \qquad k_i = \lambda_i+n-i
% \end{equation}
% where $k_1>k_2>\cdots$ is the decreasing enumeration of a Maya diagram
% defined by \eqref{eq:tnlambda}.  If $n=\ell(\lambda)$ we write simply
% $\tau_\lambda$. Note that the above definition is just a special case
% of \eqref{eq:pWdef2}.

Given univariate polynomials $p_1(z),\ldots, p_n(z)$, define the
multivariate functions
\begin{align}
  &\Delta[p_1,\ldots, p_n](z_1,\ldots, z_n)  =
    \begin{vmatrix}
      p_1(z_1) & p_1(z_2) & \ldots & p_1(z_n)\\
      p_2(z_1) & p_2(z_2) & \ldots & p_2(z_n) \\
      \vdots & \vdots & \ddots & \vdots \\
      p_n(z_1) & p_n(z_2) & \ldots & p_n(z_n) 
    \end{vmatrix}\\
%  &V(z_1,\ldots, z_n) =\prod_{1\leq i<j\leq n} (z_i-z_j)\\
  \label{eq:Spdef}
    % \frac{\det(p_i(z_j))_{i,j=1}^n }{\Delta[\fm_{n-1},\ldots, \fm_1,
    % \fm_{0}]}\intertext{where}
  &S[p_1,\ldots,p_n] =
    \frac{\Delta[p_1,\ldots, p_n] }{\Delta[\fm_{n-1},\ldots, \fm_1,
    \fm_{0}]}\intertext{where}
               &\fm_k(z) = z^k
\end{align}
is the $k\supth$ degree monomial function.  Thus,
\[ \Delta[\fm_{n-1},\ldots,    \fm_{0}](z_1,\ldots, z_n) =
  \begin{vmatrix}
    z_1^{n-1} &  \ldots & z_n^{n-1} \\
    \vdots &  \ddots & \vdots \\
    1  & \ldots & 1
  \end{vmatrix} = \prod_{1\leq i<j\leq n} (z_i-z_j) \] is the usual
Vandermonde determinant, while $S[p_1,\ldots, p_n]$ is a symmetric
polynomial in $z_1,\ldots, z_n$.
% Let $\Wr[p_1,\ldots, p_n]$ denote the usual Wronskian operator.

Let $\lambda$ be a partition. For $n\geq \ell(\lambda)$, let
$t_1> t_2> \cdots$ be the decreasing enumeration of $M^{(n)}_\lambda$
as per \eqref{eq:tnlambda}. The $n$-variate Schur polynomial is the
symmetric polynomial
\[ \fs^{(n)}_\lambda = S[\fm_{t_1}, \ldots, \fm_{t_n}].\] The Schur
polynomial $\fs^{(n)}_\lambda$ is the character of the irreducible
representation of the general linear group $\GL_n$ corresponding to
partition $\lambda$.  Moreover, the Weyl dimension formula asserts
that
\begin{equation}
  \label{eq:wdf}
  \fs^{(n)}_\lambda(1,\ldots, 1) =  \prod_{1\leq i<j\leq n}\frac{
    \lambda_i-\lambda_j   +j-i}{j-i}
  = \left(\prod_{j=1}^{n-1} j!\right)^{-1}\!\!\!\!\prod_{1\leq i<j\leq
    n}    (t_i-t_j) 
\end{equation}
is the dimension of the representation in question.  

For $n\geq 1$, let
\[ \fh^{(n)}_k(z_1,\ldots, z_n) = \sum_{1\leq i_1 \leq i_2\leq \cdots
    \leq i_k \leq n} \!\!\!\!\!\!z_{i_1}z_{i_2} \cdots z_{i_k},\quad
  k=1,2,\ldots \] denote the complete symmetric polynomial of degree
$k$ in $n$ variables.  These polynomials may also be defined by means
of the generating function
\begin{equation}
  \label{eq:hgf}
  \sum_{k=0}^\infty \fh^{(n)}_k(z_1,\ldots, z_n)u^k = \prod_{i=1}^n
  \frac{1}{1-z_i u}.
\end{equation}
The classical Jacobi-Trudi identity is a determinental representation
of the Schur polynomials in terms of complete symmetric polynomials.
\begin{proposition}
  Let $\lambda$ a partition and $n\geq \ell(\lambda)$ we have
  \begin{equation}
    \label{eq:fslambdadet}
    \fs^{(n)}_\lambda = \det \left( \fh^{(n)}_{\lambda_i+j-i}
    \right)_{i,j=1}^{\ell(\lambda)}.  
  \end{equation}
\end{proposition}

We now describe a closely related identity based on symmetric power
functions.  Define the ordinary Bell polynomials
$\fB_k(t_1,\ldots, t_k),\; k=0,1,2,\ldots$ by means of the power
generating function
\begin{equation}
  \label{eq:Bgf}
  \exp\left(\sum_{k=0}^\infty t_k u^k\right) = \sum_{k=0}^\infty
  \fB_k(t_1,\ldots, t_k) u^k.  
\end{equation}
Since
\[ \exp\left(\sum_{j=0}^\infty t_k u^k\right) = \sum_{j=0}^\infty
  \frac{1}{j!}\left(\sum_{k=0}^\infty t_k u^k\right)^j ,\] the multinomial
formula implies that,
\[
  \begin{aligned}
    \fB_k(t_1,\ldots, t_k) &= \sum_{j_1,\ldots, j_\ell\geq 0 \atop
      \Vert j\Vert=n} \frac{t^{j_1}_1}{j_1!} \frac{t^{j_1}_2}{j_2!}
    \cdots \frac{t^{j_\ell}_\ell}{j_\ell!},\qquad \Vert j \Vert = j_1
    + 2j_2
    + \cdots + \ell j_\ell\\
    &= \frac{t_1^k}{k!} + \frac{t_1^{k-2}t_2}{(k-2)!} + \cdots +
    t_{k-1} t_1 + t_k
  \end{aligned}
\]
The Bell polynomials are instrumental in describing the relation
between complete homogeneous polynomials and symmetric power
polynomials.  For a given $n\geq 1$, let
\[ \fp^{(n)}_k(z_1,\ldots, z_n) = \sum_{j=1}^n z_j^k \] denote the
symmetric $k\supth$ power polynomial in $n$ variables. These
polynomials admit the following generating function
\begin{equation}
  \label{eq:pgf}
  - \log  \prod_{j=1}^n (1-z_iu) = \sum_{j=1}^\infty \fp^{(n)}_j(z_1,\ldots,
  z_n) \frac{u^j}{j}.
\end{equation}
Comparing the generating functions \eqref{eq:hgf} \eqref{eq:Bgf} \eqref{eq:pgf}
yields the following identity
\begin{equation}
  \label{eq:hkBkp}
  \fh^{(n)}_k = \fB_k\left(\fp^{(n)}_1, \frac12 \fp^{(n)}_2,\ldots,
    \frac1k \fp^{(n)}_k\right). 
\end{equation}

With these preliminaries out of the way we can present the following
alternative version of the Jacobi-Trudi formula
\cite{veselovsergeev,harnadlee}.  For a partition $\lambda$, the
define the Schur function $\fS_\lambda$ to be the multivariate
polynomial
\begin{equation}
  \label{eq:fSlambdadef}
  \fS_\lambda =  \det(
  \fB_{\lambda_i+j-i})_{i,j=1}^{\ell(\lambda)}. 
\end{equation}
Relation \eqref{eq:fslambdadet} may now be restated as
\begin{equation}
  \label{eq:fsfSfp}
  \fs^{(n)}_\lambda = \fS_\lambda(\fp^{(n)}_1, \frac12 \fp^{(n)}_2 ,
  \ldots , \frac{1}{k} \fp^{(n)}_k,\ldots).
\end{equation}
Relation \eqref{eq:fsfSfp} will be instrumental in the proof of the
following.
\begin{theorem}
  \label{thm:tauschur}
  Let $\lambda$ be a partition.  Then,
  \[ \tau_\lambda(z) = C_\lambda\fS_\lambda(2z,-1),\]
  where
  \[ C_\lambda = 2^{n(n-1)/2} \prod_{j=1}^n (\lambda_i +
    n-i)!,\qquad n= \ell(\lambda).\]
\end{theorem}
\begin{proof}
  Specializing \eqref{eq:Bgf} and using \eqref{eq:hermgf}, we observe
  that
  \begin{align*}
    \sum_{k=0}^\infty  \fB(2z,-1)u^k
    &= \exp\left(2z u -  u^2/2\right) =  \sum_{k=0}^\infty H_n(z)
      \frac{u^k}{n!}       
  \end{align*}
  Hence,
  \[   \fB_n(2z,-1) = \frac{H_n(z)}{n!} ,\quad n=0,1,2,\ldots, \]
  Hence, by \eqref{eq:fSlambdadef},
  \begin{align*}
    \fS_\lambda(2z,-1) &= \det \left( \frac{H_{\lambda_i +
                         j-i}(z)}{(\lambda_i + j-i)!}
                         \right)_{i,j=1}^{\ell(\lambda)}. 
  \end{align*}
  Hence, by the identity
  \[ H_n' = 2n H_{n-1},\quad n=1,2,\ldots \]
  we have
  \[ H_{t_i}^{(n-j)}=\frac{t_i!}{(\lambda_i+j-i)!} 2^{n-j} H_{\lambda_i+j-i}
    ,\quad i,j=1,\ldots, n, \]
  where as above,
  \[ t_i = \lambda_i + n - i,\qquad n= \ell(\lambda).\] Hence, from
  the definition \eqref{eq:taulambda}, we have
  \[ \tau_\lambda = \det \left(  2^{n-j} t_i! \frac{H_{\lambda_i +
          j-i}(z)}{(\lambda_i + j-i)!}\right)
    = C_\lambda
    \fS_\lambda(2z,-1) \]
\end{proof}

The following results will lead to yet another description of the
Hermite-type $\tau$-function, one that is related to the Boson-Fermion
correspondence \cite{grandati2014schur} (although we do not discuss
this here).
\begin{proposition}
  Let $p_1(z),\ldots, p_n(z)$ be polynomials. Then,
  \begin{equation}
    W[p_1,\ldots, p_n](z) = \left(\prod_{j=1}^{n-1} j!\right)
    S[p_1,\ldots,p_n](z,\ldots, z)
  \end{equation}
\end{proposition}
\begin{proof}
  Express the given
  polynomials as
  \[ p_i = \sum_{j=0}^{\infty} p_{ij} \fm_j,\quad p_{ij} \in
    \C. \] Note: the above sum is actually finite, because $p_{ij} =0$
  for $j$ sufficiently large.  Hence,
  \begin{align*}
    \Wr[p_1,\ldots,p_n]
    &=\sum_{t_1,\ldots, t_n = 0}^\infty \left(\prod_{i=1}^n p_{it_i}\right)
      \Wr[\fm_{t_1} ,\ldots, \fm_{t_n}]\\
    &= \sum_{t_1>\cdots > t_n\ge0} \sum_{\pi\in \cS_n}
     \sgn(\pi) \left(\prod_{i=1}^n p_{it_{\pi_i}}\right)
            \Wr[\fm_{t_1} ,\ldots, \fm_{t_n}]
  \end{align*}
  where $\cS_n$ is the group of permutations of $\{ 1,\ldots, n\}$.

  By an elementary calculation,
  \[\Wr[\fm_{t_1} ,\ldots,    \fm_{t_n}] =\prod_{1\leq i<j\leq
      n}\!\!\!\!  (t_j-t_i)\; \fm_{|\lambda|}\] Thus, by
  introducing the abbreviation
  \[ p_\lambda =  \sum_{\pi\in \cS_n}
    \sgn(\pi)\prod_{i=1}^n  
    p_{it_{\pi_i}} \] and making use of \eqref{eq:wdf} we may write
  \begin{equation}
    \label{eq:Wp1pell}
    \Wr[p_1,\ldots, p_n] = \prod_{j=1}^{n-1} j! \sum_{\ell(\lambda)\le n}
    \fs_\lambda(1,\ldots, 1 )p_\lambda \fm_{|\lambda|} 
  \end{equation}
  Of course, the sum is actually finite because $p_\lambda=0$ if
  $\lambda_1 > \max \{ \deg p_1,\ldots, \deg p_n\}$.

  By the multi-linearity and skew-symmetry of the determinant
  \eqref{eq:Spdef},
  \begin{align*}
    S[p_1,\ldots, p_n]
    &= \sum_{t_1,\ldots, t_n=0}^\infty 
      \left(\prod_{i=1}^n p_{it_i}\right)  S[\fm_{t_1},\ldots,
      \fm_{t_n}]\\
    &= \sum_{t_1>\cdots > t_n\ge0} \sum_{\pi\in \cS_n}
     \sgn(\pi) \left(\prod_{i=1}^n p_{it_{\pi_i}}\right)S[\fm_{t_1},\ldots,
      \fm_{t_n}] \\
    &= \sum_{\ell(\lambda)\leq n} p_\lambda \fs_\lambda.
  \end{align*}
  Since $\fs_\lambda$ is a homogeneous polynomial whose total degree is
  equal to $|\lambda|$, we have
  \begin{equation}
    \fs_\lambda(z,\ldots, z) =   \fs_\lambda(1,\ldots, 1) z^{|\lambda|}.
  \end{equation}
  Hence,
  \[ S[p_1,\ldots, p_n](z,\ldots, z) = \sum_\lambda p_\lambda
    \fs_\lambda(1,\ldots, 1) z^{|\lambda|}.\]
  The desired
  conclusion now follows directly by \eqref{eq:Wp1pell}.
\end{proof}
\begin{corollary}
  For  $0\leq t_q < \cdots < t_1$ we have
  \[ \tau(t_q,\ldots, t_1)(z) = \left(\prod_{j=1}^{n-1}
      j!\right)S[H_{t_q},\ldots, H_{t_1}](z,\ldots, z). \]
\end{corollary}
% \begin{itemize}
% \item Define complete symmetric polynomials and state the classical
%   Jacobi-Trudy formula which gives $\fs_\lambda$ as a determinant of
%   complete symmetric polynomials.
% \item Define Schur functions in terms of Sato-type generating
%   function.
% \item Re-express Schur polynomials as schur functions composed with
%   symmetric powers.
% \item Put the pieces together by expressing Hermite Wronskians as
%   specializations of the Schur functions.  This uses the Hermite
%   generating function.  The boson-fermion formula expresses hermite
%   Wronskians as a Jacobi-Trudi determinant.   Main result is that
%   Hermite Wronskians are specializations of the Schur functions.
% \end{itemize}

\section{Hermite-type rational solutions}
In this Section we develop the relationship between cyclic dressing
chains and Hermite-type $\tau$-functions. Every element of the chain
will be represented by a Maya diagram with successive elements related
by flip operations. With this representation, the construction of
rational solutions to a cyclic dressing chains is reduced to a
combinatorial question regarding cyclic Maya diagrams.

In what follows we make use of the Hirota bilinear notation:
\begin{align}
  D f \cdot g &=  f' g - g'f \\
%  (E f\cdot g)(z) &= z\, (D f\cdot g)(z),\\
  D^2 f\cdot g &= f'' g -  2 f' g' + g'' f
\end{align}
A direct calculation then establishes the following.
\begin{proposition}
  \label{prop:tauwU}
  Let $f=f(z),g=g(z)$ be rational functions, and let
  \begin{align}
    \label{eq:wzfg}
     w &= -z- \frac{f'}{f} +
            \frac{g'}{g}, \\
    U &= z^2 - 2 (\log f)'',\\
    V &= z^2 - 2 (\log g)'',
  \end{align}
  where $w=w(z), U=U(z), V=V(z)$.
  Then,
  \begin{equation}
    \label{eq:taubilin}
  \begin{aligned}
    (D^2-2zD) f\cdot g
    &= (w^2 + w' +1-U)fg\\
    &= (w^2 - w' -1-V)fg
  \end{aligned}
  \end{equation}
  % \begin{aligned}
  %   L_M -1  &= (D+w)(-D+w)+\lambda, \\
  %   L_{M'}+1 &= (D-w)(-D-w)+\lambda
  % \end{aligned}
\end{proposition}
% \begin{proof}
%   % Dividing the left-side of \eqref{eq:taubilin} through by $fg$ gives
%   % \[ \frac{f''(z)}{f(z)} - 2
%   %   \frac{f'(z)}{f(z)}\frac{g'(z)}{g(z)}
%   %   +\frac{g''(z)}{g(z)}+ 2 z
%   %   \left(\frac{f'(z)}{f(z)} -
%   %     \frac{g'(z)}{g(z)}\right). \]
%   A direct calculation shows that the above is equivalent to both
%   \eqref{eq:UMw} and \eqref{eq:UM'w}.
% \end{proof}

We are now able to exhibit a bilinear formulation for the dressing
chain \eqref{eq:wachain}.

\begin{proposition}
  \label{prop:tauchain}
  Suppose that
  $\tau_i=\tau_i(z),\epsilon_i,\sigma_i\in \{ -1,1\},\; i=0,1,\ldots,
  p-1$ is a sequence of functions and constants that satisfies
  \begin{equation}
    \label{eq:tauchain}
    (D^2+2\sigma_i zD+\epsilon_i) \tau_i \cdot \tau_{i+1}=0,\qquad
    i=0,1,\ldots, p-1      \mod p
  \end{equation}
  Then,
  \begin{align}
    \label{eq:witaui}
    w_i &= \sigma_iz- \frac{\tau_{i}'}{\tau_{i}} +
      \frac{\tau_{i+1}'}{\tau_{i+1}}, \qquad
    i=0,1,\ldots, p-1  \mod p\\ \nonumber
    a_i &= \epsilon_i- \epsilon_{i+1}+\sigma_i+\sigma_{i+1}
  \end{align}
  satisfy the $p$-cyclic dressing chain \eqref{eq:wachain} with
  \begin{equation}
    \label{eq:Deltasigma}
     \Delta = -2 \sum_{j=0}^{p-1} \sigma_i.
  \end{equation}
\end{proposition}
\begin{proof}
  Set
  \begin{align*}
    U_i(z) &= z^2 - 2 (\log \tau_i)''(z) -2
             \sum_{j=0}^{i-1}\sigma_j,\quad i=0,1,2,\ldots, p \\
    \lambda_i &=\epsilon_i-2\sum_{j=0}^{i-1} \sigma_j -\sigma_i.
  \end{align*}
  Observe that
  \[ (D^2-2\sigma_izD + \lambda) \tau_{i+1} \cdot \tau_i =
    (D^2+2\sigma_i zD +\lambda) \tau_i \cdot \tau_{i+1}.\]
  Hence, by \eqref{eq:taubilin},
  \begin{align*}
    w_i^2+w_i'+ \epsilon_i
    &= z^2 - 2 (\log \tau_i)'' + \sigma_i\\
    w_{i}^2-w_{i}'+ \epsilon_i
    &= z^2 - 2 (\log      \tau_{i+1})'' - \sigma_i 
  \end{align*}
  The above relations are equivalent to \eqref{eq:Riccati}, and hence
  the $U_i(z),\lambda_i$ constitute a periodic factorization chain
  \eqref{eq:Dxform} with $\Delta=U_p-U_0$ given by
  \eqref{eq:Deltasigma}.  Applying \eqref{eq:alphaidef}, we obtain
  \[ a_i = \lambda_i-\lambda_{i+1} = \epsilon_i-\epsilon_{i+1} +
    \sigma_i + \sigma_{i+1}.\] Hence, with \eqref{eq:witaui} as
  definition of $w_i(z)$ and $a_i$, relation \eqref{eq:wachain} is
  satisfied.
\end{proof}
\noindent
Note: this proposition should not be taken as a claim that \emph{all}
rational solutions of \eqref{eq:wachain} may be obtained in this
fashion.

In order to obtain polynomial solutions of \eqref{eq:tauchain} we now
introduce the following.

\begin{definition}
  A flip at $m\in \Z$ is the involution $\phi_m:\cM\to \cM$ defined by
\begin{equation}\label{eq:flipdef}
 \phi_m : M \mapsto
\begin{cases}
   M \cup \{ m \} & \text{ if } m\notin M \\
   M \setminus \{ m \} & \text{ if } m\in M 
\end{cases},\qquad M\in \cM.
\end{equation}
In the first case, we say that $\phi_m$ acts on $M$ by a
state-deleting transformation ($\emptybox\to \boxdot$).  In the second
case, we say that $\phi_m$ acts by a state-adding transformation
($\boxdot\to\emptybox$).  We define the \emph{flip group} $\cF$ to be
the group of transformations of $\cM$ generated by flips
$\phi_m,\; m\in \Z$.  A \emph{multi-flip} is an element of $\cF$.
\end{definition}

\begin{theorem}
  \label{thm:Mbilin}
  Let $M_1\subset \Z$ be a Maya diagram, and
  $M_2 = M_1 \cup \{ m\},\; m\notin M_1$ another Maya diagram obtained
  by a state-deleting transformation.  Then, the corresponding
  pseudo-Wronskians satisfy the bilinear relation
  \begin{equation}
    \label{eq:bilinM12}
    (D^2-2zD + \epsilon) \tau_{M_1} \cdot \tau_{M_2} = 0,\quad \tau_M=\tau_M(z)
   \end{equation}
   where
   \begin{equation}
     \label{eq:epsilonM12}
     \epsilon = 2 (\deg \tau_{M_2} - \deg \tau_{M_1}).
   \end{equation}
   Conversely, suppose that \eqref{eq:bilinM12} holds for some Maya
   diagrams $M_1, M_2\subset \Z$ and some $\epsilon\in \C$.  Then,
   necessarily $M_2$ is obtained from $M_1$ by a state-deleting
   transformation and $\epsilon$ takes the value shown above.
\end{theorem}
\noindent The proof of this theorem requires a number of intermediate
results.
\begin{lemma}
  \label{lem:tau0123}
  Let $\tau_0(z),\tau_1(z), \tau_2(z), \tau_3(z)$ be rational
  functions such that
  \begin{equation}
    \label{eq:tau3def}
    \tau_0 \tau_2 = \Wr[\tau_1,\tau_3].
  \end{equation}
  Associate the edges of the following diagram to the bilinear
  relations displayed below:
  
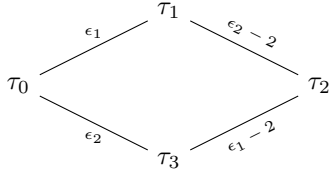
\begin{center}
    \begin{tikzpicture}
      \node (t3) at (0:2) {$\tau_2$};
      \node (t1) at (90:1) {$\tau_1$}
      edge node[above,sloped] {\tiny$\epsilon_2-2$}   (t3);
      \node (t2) at (270:1) {$\tau_3$}
      edge node[below,sloped] {\tiny$\epsilon_1-2$} (t3);
      \node (t0) at      (180:2) {$\tau_0$}
      edge  node[above] {\tiny$\epsilon_1$}  (t1)
      edge  node[below] {\tiny$\epsilon_2$}  (t2);
    \end{tikzpicture}
  \end{center}
  \begin{align}
    \label{eq:tau01}
    &(D^2-2zD+\epsilon_1) \tau_0 \cdot \tau_1 = 0,\\
    \label{eq:tau03}
    &(D^2-2zD+\epsilon_2) \tau_0 \cdot \tau_3 = 0,\\
    \label{eq:tau12}
    &(D^2-2zD+\epsilon_2-2) \tau_1 \cdot \tau_2 = 0, \\
    \label{eq:tau32}
    &(D^2-2zD+\epsilon_1-2) \tau_3 \cdot \tau_2 = 0,
  \end{align}
  where $\epsilon_1,\epsilon_2\in \C$ are constants.
  Then, necessarily any two relations corresponding to connected edges
  entail the other two relations.
\end{lemma}
\begin{proof}
  The lemma asserts two claims.  First that \eqref{eq:tau01}
  \eqref{eq:tau03} together, are logically equivalent to
  \eqref{eq:tau12} \eqref{eq:tau32} together.  The other assertion is
  that \eqref{eq:tau01} \eqref{eq:tau12} together are logically
  equivalent to \eqref{eq:tau03} \eqref{eq:tau32} together.  We will
  demonstrate how \eqref{eq:tau01} \eqref{eq:tau03} together imply
  \eqref{eq:tau12}.  All the other demonstrations can be argued
  analogously, and so we omit them.

  Begin by setting
  \begin{align}
    \label{eq:wijdef}
    w_{ij}(z) &=  - z - \frac{\tau_i'(z)}{\tau_i(z)} +
                \frac{\tau_j'(z)}{\tau_j(z)} ,\quad i\in \{
                0,1,2,3\},\; i\neq j,\\
    U_i(z) &= z^2 - 2 (\log \tau_i)''(z),\quad i\in \{ 0,1,2,3\}.    
  \end{align}
  By \eqref{eq:taubilin} of Proposition \ref{prop:tauwU},
  \[ U_0 = w_{01}^2 + w_{01}' + \epsilon_1+1 = w_{03}^2+ w_{03}' +
    \epsilon_2+1,\]  
  which we rewrite as
  \begin{equation}
    \label{eq:w0201}
     \frac{w'_{03} - w'_{01}}{w_{03}-w_{01}} =  -w_{03}-w_{01} +
    \frac{\alpha}{w_{03} - w_{01}},\quad \alpha = \epsilon_1 - \epsilon_2.    
  \end{equation}
  Write
  \[ \Wr[\tau_1, \tau_3 ] = \tau_1 \tau_3 \left(
      \frac{\tau_3'}{\tau_3} - \frac{\tau_1'}{\tau_1}\right) = \tau_1
    \tau_3 (w_{03}- w_{01}).\]
  Hence, by  \eqref{eq:tau3def},
  \begin{gather*}
    \tau_0 \tau_2 = \tau_1 \tau_3 (w_{03}- w_{01})\\ 
     \frac{\tau_2'}{\tau_2} =  - \frac{\tau_0'}{\tau_0} +
    \frac{\tau_1'}{\tau_1}+
    \frac{\tau_3'}{\tau_3}  +\frac{ w_{03}' - w_{01}'}{w_{03}-w_{01}} 
  \end{gather*}
  which by \eqref{eq:wijdef} \eqref{eq:w0201} maybe rewritten as
  \begin{equation}
    \label{eq:w13}
    \begin{aligned}
      &w_{12} = w_{03} + \frac{w_{03}' - w_{01}'}{w_{03}-w_{01}} =
      -w_{01} + \frac{\alpha}{w_{03}- w_{01}},\\
      &(w_{01}+w_{12})(w_{03}-w_{01}) = \alpha.      
    \end{aligned}
  \end{equation}
  It follows that
  \[ \frac{w_{12}'+w_{01}'}{w_{12} + w_{01}} + \frac{w_{03}' -
      w_{01}'}{w_{03}-w_{01}} =0\]
  Hence,
  \begin{align*}
    &w_{12} - w_{01}  + \frac{w_{12}' + w_{01}'}{w_{12} + w_{01}} -
    \frac{\alpha}{w_{12} + w_{01} } =\\
    &\qquad = w_{03}+ \frac{w_{03}' -
    w_{01}'}{w_{03}-w_{01}}-w_{01} + \frac{w_{12}' + w_{01}'}{w_{12} + w_{01}}
    - \frac{\alpha}{w_{12} + w_{01} }    = 0
  \end{align*}
  Equivalently,
  \begin{equation}
    \label{eq:w13w01}
    w_{12}^2 - w_{01}^2 + w_{12}' + w_{01}' - \alpha = 0.
  \end{equation}
  Hence, by  \eqref{eq:taubilin} 
  \[ U_1= w_{01}^2 - w_{01}' + \epsilon_1-1  = w_{12}^2 + w_{12}' + \epsilon_2-1.\]
  Relation \eqref{eq:tau12} now follows by Proposition  \ref{prop:tauwU}.  
\end{proof}
\begin{proof}[Proof of Theorem \ref{thm:Mbilin}]
  Suppose that $M_2= M_1\cup \{m\},\; m\notin M_1$.  We claim that
  \eqref{eq:bilinM12} holds. By Proposition \ref{prop:equiv} no
  generality is lost if we assume that $M_1$ is in standard form, and
  hence that
  \[ \tau_{M_1} = \tau(t_q,\ldots t_1),\quad t_q< \cdots < t_1 \] is a
  pure Wronskian.  We proceed by induction on $q$, the number of
  positive elements of $M_1$.  If $q=0$, then $\tau_{M_1}=1$ and
  $\tau_{M_2} = H_m$.  The bilinear relation \eqref{eq:bilinM12} is
  then nothing but the classical Hermite differential equation
  \[ H_m''(z) - 2z H_m'(z) + 2m H_m(z) = 0.\] Suppose that $q>0$ and
  that the claim has been shown to be true for all $M_1$ with fewer
  positive elements.  Set
  \[ M_0 = M_1 \setminus \{ t_q\},\qquad M_2 = M_0 \cup \{ m \}.\]
  By an elementary calculation,
  \begin{align*}
    \deg \tau_{M_1} &=  \sum_{i=1}^q m_i - \frac12 q(q-1)\\
    \deg \tau_{M_2} &=  \sum_{i=1}^q m_i + m -  \frac12 q(q+1)\\
    \deg \tau_{M_0} &=  \sum_{i=1}^{q-1} m_i - \frac12 (q-1)(q-2)\\
    \deg \tau_{M_3} &=  \sum_{i=1}^{q-1} m_i + m - \frac12 q(q-1)
  \end{align*}
  Hence, by the inductive hypothesis, \eqref{eq:tau01}
  \eqref{eq:tau03} hold with
  \begin{align*}
    \epsilon_1 &= 2(\deg \tau_{M_1} - \deg \tau_{M_0}) = 2m_q-2q+2,\\
    \epsilon_2 &= 2(\deg \tau_{M_3} - \deg \tau_{M_0}) = 2m-2q+2   
  \end{align*}
  Observe that
  \[ 2(\deg \tau_{M_2} - \deg \tau_{M_1}) = 2m-2q\]
  Hence, by Lemma \ref{lem:tau0123} \eqref{eq:bilinM12} holds also.

  Conversely, suppose that \eqref{eq:bilinM12} holds for some
  $\lambda\in \C$.  We claim that
  $M_2 = M_1 \cup \{m\},\; m\notin M_1$.
  Without loss of generality, suppose that
  \[ M_1 = (\;\mid t_q, \ldots, t_1    ),\quad t_1>\cdots > t_q>0 \]
  is in standard form and hence that
  \[ \tau_{M_1} = \tau(t_q,\ldots, t_1).\] Using Proposition
  \ref{prop:equiv}, assume without loss of generality that
  \[ \tau_{M_2} = \tau(\hatt_{\hq},\ldots, \hatt_1),\] is also
  a pure Wronskian with $\hq \geq q$.  Using the reduction argument
  above it then suffices to demonstrate this claim for the case where
  $M_1$ is the trivial Maya diagram; $q=0$.  In this case,
  \eqref{eq:bilinM12} reduces to the Hermite differential equation.
  The Hermite polynomials are the unique polynomial solutions and so
  the claim follows.
\end{proof}

We see thus that Hermite-type $\tau$-functions are indexed by Maya
diagrams, and that flip operations on Maya diagrams correspond to
bilinear relations between these $\tau$-functions.  Since $p$-cyclic
chains of bilinear relations \eqref{eq:tauchain} correspond to
rational solutions of the $A_{p-1}$ Painlev\'e system, it now becomes
feasible to construct rational solutions in terms of cycles of Maya
diagrams.

\begin{definition}
  \label{def:multiflip}
  For $p=1,2,\ldots$ and $\bmu=(\mu_0,\ldots, \mu_{p-1})\in \Z^p$, let
   \begin{equation}
     \label{eq:phimudef}
     \phi_\bmu= \phi_{\mu_0} \circ \cdots \circ
     \phi_{\mu_{p-1}},\quad
   \end{equation}
   denote the indicated multi-flip.
   % For a $p$-tuple $\bmu\in \Z^p$ let
   % $m_\bmu(a),\; a\in \Z$ denote the number of times that $j$ occurs
   % in the tuple, and let
   % \[ \supp\bmu = \{ a\in \Z \colon m_\bmu(a)>0 \} \] be the set
   % consisting of all the elements of the tuple.
   We will call $\bmu\in \Z^p$ non-degenerate if the set
   $\{ \mu_0,\ldots, \mu_{p-1} \}$ has cardinality $p$, that is if
   $\mu_i\neq \mu_j$ for $i\neq j$.  If $\hbmu\subset \Z$ is a set of
   cardinality $q$, we let $\phi_\hbmu = \phi_\bmu$, where
   $\bmu\in \Z^q$ is any non-degenerate enumeration of $\hbmu$.
   Finally, given a finite set $\hbmu\subset \Z$ we will call a
   sequence $\bmu=(\mu_0,\ldots, \mu_{p-1})$ an \emph{odd enumeration}
   of $\hbmu$ if
   \[ \hbmu = \{ a\in \Z \colon m_\bmu(a) \equiv 1 \mod 2 \},\] where
   $m_\bmu(a)\geq 0$ is the number of times that $a\in \Z$ occurs
   in $\bmu$.
\end{definition}

% \begin{definition}
%   For $\bmu\in \Z^p$ we define the \emph{odd support} $\supp_1(\bmu)$ to be
%   the set
%   \[ \supp_1 \bmu = \{ a \in \Z \colon m_\bmu(a) \equiv 1 \mod 2 \}\]
%   We will say that 
%   $\bmu\in \Z^p$ is an \emph{odd enumeration} of  a set
%   $\hbmu\subset \Z$ if $\hbmu= \supp_1(\bmu)$.
% \end{definition}
\begin{proposition}
  \label{prop:oddenum}
  For $\bmu\in \Z^p$ and a finite $\hbmu\subset \Z$ we have
  $\phi_\bmu = \phi_\hbmu$ if and only if $\bmu$ is an odd enumeration
  of $\hbmu$.
\end{proposition}
We are now ready to introduce the basic concept of this section.

\begin{definition}
  We say that $M$ is $p$-cyclic with shift $k$, or $(p,k)$ cyclic, if
  there exists a $\bmu \in \Z^p$ such that
  \begin{equation}
    \label{eq:cyclicMdef}
    \phi_\bmu(M) = M+k.
  \end{equation}
  We will say that $M$ is $p$-cyclic if it is $(p,k)$ cyclic for some
  $k\in \Z$.
\end{definition}

\begin{proposition}
  \label{prop:muM1M2}
  For Maya diagrams $M,M'\in \cM$, define the set
  \begin{equation}
    \label{eq:muM1M2}
    \Upsilon(M,M') = (M \setminus M') \cup (M'\setminus M)
  \end{equation}
  Then $\phi_\hbmu$ where $\hbmu=\Upsilon(M,M')$ is the unique
  multi-flip such that $ M' = \phi_{\hbmu}(M)$ and
  $M = \phi_{\hbmu}(M')$.
\end{proposition}
As an immediate corollary, we have the following.
\begin{proposition}
  Let $k$ be an integer, $M\in \cM$ a Maya diagram, and let $p$ be the
  cardinality of $\Upsilon(M,M+k)$. Then $M$ is $(p+2j,k)$-cyclic for
  every $j=0,1,2,\ldots$.
\end{proposition}
\begin{proof}
  Let $\hbmu = \Upsilon(M,M+k)$.  Then $\phi_{\hbmu}(M) = M+k$, by the
  preceding Proposition, and hence $M$ is $(p,k)$ cyclic .  Let
  $\bmu\in \Z^p$ be an odd enumeration of $\hbmu$; i.e., $\bmu$ is
  obtained by adjoining $j$ pairs of repeated indices to the elements
  of $\hbmu$. By Proposition \ref{prop:oddenum},
  \[\phi_\bmu(M) = \phi_\hbmu(M) = M+k,\] and therefore $M$ is also
  $(p+2j,k)$-cyclic.
\end{proof}
% \begin{definition}
%   Given a $(p,k)$ cyclic Maya diagram $M$ and a corresponding
%   multi-flip $\bmu$ satisfying \eqref{eq:cyclicMdef} we will refer to
%   the sequence of Maya diagrams
%  as a Maya $p$-cycle generated by $\bmu$.
% \end{definition}

The following result should be regarded as a refinement of Proposition
\ref{prop:tauchain}.
\begin{proposition}
  \label{prop:Mtauw}
  Let $M\in \cM$ be a Maya diagram, $k$ a non-zero integer, and
  $(\mu_0,\ldots, \mu_{p-1})\in \Z^p$ an odd enumeration of
  $\Upsilon(M,M+k)$. Extend $\bmu$ to an infinite $p$-quasiperiodic
  sequence by letting
  \[ \mu_{i+p} = \mu_i+k,\quad, k=0,1,2,\ldots \]
  and recursively define
  \begin{equation}
    \label{eq:Mchain}
    \begin{aligned}
      M_0 &=M,\\
      M_{i+1} &= \phi_{\mu_i}(M_{i}),\quad    i=0,1,2,\ldots
    \end{aligned}
  \end{equation}
  so that $M_{i+p} = M_i+k$ by construction.  Next, for
  $i=0,1,\ldots$, let
  \begin{equation}
    \label{eq:HM2w} 
    \begin{aligned}
      \tau_i &= \tau_{M_i},\\
      a_i&=2(\mu_i-\mu_{i+1}),\\
      \sigma_i&=\begin{cases}
        +1 & \textit{ if } \, \mu_i\in M_{i} \\
        -1 & \textit{ if } \, \mu_i\in M_{i+1} 
      \end{cases},\\
      w_i(z)&= \sigma_{i} \,z+ \frac{\tau_{i+1}'(z)}{\tau_{i+1}(z)}-
      \frac{\tau_i'(z)}{\tau_i(z)},
    \end{aligned}
  \end{equation}
  which are all $p$-periodic by construction.  The just-defined
  $w_i(z), a_i,\; i=0,1,\ldots, p-1\mod p$ constitute a rational
  solution to the $p$-cyclic dressing chain \eqref{eq:wachain} with
  shift $\Delta=2k$.
\end{proposition}
\begin{proof}
  Set
  \[ \epsilon_i = 2\sigma_i (\deg \tau_{i} - \tau_{i+1}),\quad
    i=0,1,\ldots p-1\mod p.\] Theorem \ref{thm:Mbilin} then implies
  that
  \[ (D^2+2\sigma_iz D + \epsilon_i) \tau_i \cdot \tau_{i+1} = 0.\]
  Proposition \ref{prop:equiv} allows us to assume, without loss of
  generality, that
  \[ M_i = M(\mid t_{i,q_i},\ldots, t_{i,1}) \] and that
  $\mu_i \geq 0$ for all $i$.  Such an outcome can can imposed by
  applying a sufficiently positive translation to the Maya diagrams in
  question, without altering the log-derivatives of the
  $\tau$-functions.  Since each $\tau_i$ is a Wronskian of polynomials
  of degrees $t_{i,1},\ldots, t_{i,q_i}$, we have
  \[ \deg \tau_i = \sum_{j=1}^{q_i} t_{i,j} - \frac12 q_i(q_i-1).\]
  If $\sigma_i=-1$, then
  \[ \{ t_{i+1,1},\ldots, t_{i+1,q_{i+1}} \} = \{ t_{i,1},\ldots ,
    t_{i,q_i} \} \cup \{ \mu_i \},\qquad q_{i+1} = q_i+1.\]
  If $\sigma_i=+1$, then
  \[ \{ t_{i,1},\ldots, t_{i+1,q_{i}} \} = \{ t_{i+1,1},\ldots ,
    t_{i+,q_{i+1}} \} \cup \{ \mu_i \},\quad q_{i+1} = q_i-1.
 \]
  It follows that
  \begin{align*}
    \epsilon_i&=2\sigma_i (\deg \tau_{i} - \deg \tau_{i+1}) = 2\mu_i -
    2q_i + 1+\sigma_i,\\
    \epsilon_i- \epsilon_{i+1} + \sigma_i + \sigma_{i+1} &= 2\mu_i-2\mu_{i+1}+2(q_{i+1}-q_i)+\sigma_i-\sigma_{i+1}+\sigma_i+\sigma_{i+1}\\
     &=  2(\mu_i-\mu_{i+1})
  \end{align*}
  Hence, the definition of $a_i$ in \eqref{eq:HM2w} agrees with the
  definition in \eqref{eq:witaui}.  Therefore,
  $w_i(z), a_i,\; i=0,1,\ldots, p-1$ satisfy \eqref{eq:wachain} by
  Proposition \ref{prop:tauchain}. Finally, by \eqref{eq:Deltasumalpha},
  \[ \Delta = -\sum_{i=0}^{p-1} a_i = 2\sum_{i=0}^{p-1}
    (\mu_{i+1}-\mu_{i}) = 2(\mu_p - \mu_0) = 2k. \]
\end{proof}

The remaining part of the construction is to classify cyclic Maya
diagrams for a given period, which we tackle next.  Under the
correspondence described by Proposition \ref{prop:Mtauw}, the reversal
symmetry \eqref{eq:reversal2} manifests as the transformation
\[ (M_0,\ldots, M_p) \mapsto (M_p,\ldots, M_0),\quad (\mu_1,\ldots,
\mu_{p}) \mapsto (\mu_{p},\ldots, \mu_1),\quad k\mapsto -k.\]
In light of the above remark, there is no loss of generality if we
restrict our attention to cyclic Maya diagrams with a positive shift
$k>0$.

%%%%%%%%%%%%%%%%%%%%%%%%%%%%%%%%%%%%%
\section{ Cyclic Maya diagrams}\label{sec:Mayacycles}
%%%%%%%%%%%%%%%%%%%%%%%%%%%%%%%%%%%%%

In this section we introduce the key concepts of \textit{genus} and
\textit{interlacing} to achieve a full classification of cyclic Maya
diagrams.

% Consider a partition $\lambda$ as the integer sequence  $\lambda_1\geq \lambda_2\geq \cdots$ with a finite number $\ell$ of non-zero elements.  
% \begin{definition}  \label{def:genuslambda}
%   We define the genus $g\geq 0$ of a partition $\lambda$ to be the
%   cardinality of the set $\# \{ \lambda_1,\ldots, \lambda_\ell\}$. In
%   other words, $g$ is the number of distinct elements in the
%   partition.
% \end{definition}
% We recall the well known correspondence between Maya diagrams and partitions \cite{gomez2016durfee}.

% \begin{proposition}
% Let $M\subset\Z$ be a Maya diagram and $m_1>m_2>\cdots$ be its
% elements ordered in decreasing order. Consider the partition defined
% by
% \begin{equation}
%   \label{eq:lambdafromM} \lambda_i = \#\{  m\notin M \colon m < m_i\}
% ,\quad i=1,2\ldots .
% \end{equation}
% The correspondence $M\mapsto \lambda$ given by
% \eqref{eq:lambdafromM} defines a bijection between the set of
% unlabelled Maya diagrams and the set of partitions.
% \end{proposition}
% Note that the partition $\lambda$ corresponding to a given Maya diagram $M$ is invariant under translations of $M$, i.e. the Maya diagrams $M$ and $M'=M+k$ have the same partition.
% We define the genus of a Maya diagram $M$ as the genus of its corresponding partition.
% 
%   Let $M$ be a Maya diagram and $\lambda$ its corresponding partition defined by \eqref{eq:lambdafromM}. The genus $g$ of $M$ is defined to be the genus of $\lambda$, as per Definition \ref{def:genuslambda}.
% \end{definition}
\begin{definition}
  For $p=1,2,\ldots$ let $\cZ_p\subset\Z^p$ denote the set of
  non-decreasing integer sequences\footnote{Equivalently, we may
    regard $\cZ_p$ as the set of $p$-element integer multi-sets. A
    multi-set is generalization of the concept of a set that allows
    for multiple instances for each of its elements. }
  $\beta_0\leq \beta_1\leq\cdots \leq \beta_{p-1}$.  of integers.  For
  $\bbeta\in \cZ_{2g+1}$ define the Maya diagram
  \begin{equation}
    \label{eq:MBi}
    \Xi(\bbeta)= (-\infty,\beta_0) \cup [\beta_1,\beta_2) \cup
    \ \cdots \cup [\beta_{2g-1},\beta_{2g})
  \end{equation}
  where
  \[ [m,n) = \{ j\in \Z \colon m\leq j < n\}\] Let
  $\hcZ_p\subset \Z^p$ denote the set of strictly increasing integer
  sequences $\beta_0<\beta_1<\cdots < \beta_{p-1}$.  Equivalently, we
  may regard a $\bbeta\in \hcZ_p$ as a $p$-element subset of $\Z$.
\end{definition}

\begin{proposition}
  Every Maya diagram $M\in \cM$ has a unique representation of the
  form $M=\Xi(\bbeta)$ where $\bbeta\in \hcZ_{2g+1}$. Moreover, $M$ is
  in standard form if and only if $\min \bbeta = 0$.
\end{proposition}
\begin{proof}
  After removal of the initial infinite $\boxdot$ segment and the
  trailing infinite $\emptybox$ segment, a given Maya diagram $M$
  consists of $2g$ alternating empty $\emptybox$ and filled $\boxdot$
  segments of variable length.  The genus $g$ counts the number of
  such pairs.
% The odd index $b_{2i-1}$ give the lengths of
%   the empty segments, while the even index $b_{2i}$ give the lengths
%   of the filled segments.  
  The even block coordinates $\beta_{2i}$ indicate the starting
  positions of the empty segments, and the odd block coordinates
  $\beta_{2i+1}$ indicated the starting positions of the filled
  segments.  See Figure~\ref{fig:genusM} for a visual illustration of
  this construction.
\end{proof}

\begin{definition}\label{def:genus}
  We call the integer $g\geq 0$ the genus of
  $M= \Xi(\bbeta),\; \bbeta\in \hcZ_p$ and
  $(\beta_0,\beta_1,\ldots, \beta_{2g})$ the block coordinates of $M$.
\end{definition}
\noindent
% Let $\cZo$ denote the set of all finite subsets of $\Z$ having odd
% carindality, and 

\begin{proposition}
  \label{prop:1cyclic}
  Let $M = \Xi(\bbeta),\; \bbeta\in \hcZ_p$ be a Maya diagram
  specified by its block coordinates.  We then have
  \[ \bbeta = \Upsilon(M,M+1).\]
\end{proposition}
\begin{proof}
  Observe that
  \[ M+1 = (-\infty, \beta_0] \cup (\beta_1,\beta_2] \cup \cdots \cup
  (\beta_{2g-1}, \beta_{2g}],\]
  where
  \[ (m,n] = \{ j\in \Z \colon m<j\leq n \}.\]
  It follows that
  \begin{align*}
    (M+1)\setminus M &= \{ \beta_0, \ldots, \beta_{2g} \}\\
    M\setminus (M+1) &= \{ \beta_1,\ldots, \beta_{2g-1} \}.
  \end{align*}
  The desired conclusion follows immediately.
\end{proof}
% \todo{there is a problem with this proposition: a Maya diagram with genus $g$ can also be $(2g+1,-1)$-cyclic: choose bullets so that $ \phi_{\bmu}(M)=M-1$. This propagates into the Theorem. Need to fix both. }

% \begin{proof}
%   Let $(\beta_0, \beta_1, \ldots, \beta_{2g})$ be the block coordinates of
%   $M$. By direct verification,
%   \begin{equation}
%     \label{eq:muM+1}
%      \phi_{\bmu}(M) = M+1 .
%    \end{equation}
% \end{proof}

Let $\cM_g$ denote the set of Maya diagrams of genus $g$.  The above
discussion may be summarized by saying that the mapping \eqref{eq:MBi}
defines a bijection $\Xi:\hcZ_{2g+1} \to \cM_g$, and that the block
coordinates are precisely the flip sites required for a translation
$M\mapsto M+1$.
% Taking the union of all these we arrive
% at the bijection $\Xi:\cZo\to \cM$.

% \begin{remark}
%   To motivate Definition \ref{def:genus}, it is perhaps more illustrative to
%   understand the visual meaning of the genus of $M$, see
%   Figure~\ref{fig:genusM}.  
% \end{remark}

\begin{figure}[h]
\begin{tikzpicture}[scale=0.6]

\draw  (1,1) grid +(15 ,1);

\path [fill] (0.5,1.5) node {\huge ...} 
++(1,0) circle (5pt) ++(1,0) circle (5pt)  ++(1,0) circle (5pt) 
++(1,0) circle (5pt) ++(1,0) circle (5pt) 
++(2,0) circle (5pt) ++(1,0) circle (5pt) 
++ (3,0) circle (5pt)  ++(1,0) circle (5pt)   ++ (1,0) circle (5pt) 
++ (3,0) node {\huge ...} +(1,0) node[anchor=west] { $M =  (-\infty,\beta_0)\cup [ \beta_1,\beta_2) \cup [ \beta_3,\beta_4)$};

\draw[line width=1pt] (4,1) -- ++ (0,1.5);

\foreach \x in {-3,...,11} 	\draw (\x+4.5,2.5)  node {$\x$};
\path (6.5,0.5) node {$\beta_0$} ++ (1,0) node {$\beta_1$}
++ (2,0) node {$\beta_2$}++ (2,0) node {$\beta_3$}++ (3,0) node {$\beta_4$}
;
\end{tikzpicture}

\caption{Block coordinates $(\beta_0,\ldots, \beta_4) = (2,3,5,7,10)$
  of a genus $2$ Maya diagrams.  Note that the genus is both the
  number of finite-size empty blocks and the number of finite-size
  filled blocks.}\label{fig:genusM}
\end{figure}
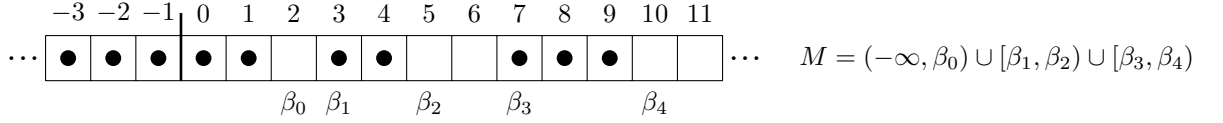

The next concept we need to introduce is the interlacing and modular
decomposition. 
\begin{definition}\label{def:interlacing}
  Fix a $k\in \N$ and let $M^{(0)}, M^{(1)},\ldots M^{(k-1)}\subset \Z$ be sets
  of integers.  We define the interlacing of these to be the set
  \begin{equation}\label{eq:interlacing} \Theta\left(M^{(0)}, M^{(1)},\ldots M^{(k-1)}\right)
    = \bigcup_{i=0}^{k-1} (k M^{(i)} +i),
 \end{equation}
 where
 \[ kM +j = \{ km + j \colon m\in M \},\quad M\subset \Z.\]
 Dually, given a set of integers $M\subset \Z$ and a $k\in \N$ define
 the sets
 \[ M^{(i)} = \{ m\in \Z \colon km+i \in M\},\quad i=0,1,\ldots, k-1.\]
 We will call the $k$-tuple of sets $\left(M^{(0)}, M^{(1)},\ldots M^{(k-1)}\right)$ the
 $k$-modular decomposition of $M$.
\end{definition}

% \begin{definition}
%   Let
%   $\bmu_i=(\mu_{i1},\ldots, \mu_{ip_i})\in \Z^{p_i},\; i=0,\ldots,
%   k-1$
%   be multi-flips.  We define the $k$-interlacing of these to be the
%   multi-flip
%   $[\bmu_0| \ldots |\bmu_{k-1}]_k \in \Z^{p_0+\cdots + p_{k-1}}$
%   defined by
%   \begin{equation} 
%     [\bmu_0|\bmu_1| \ldots | \bmu_{k-1}]_k=
%     (k\bmu_0 , k\bmu_1+1 , \cdots , k \bmu_{k-1} + k-1)
%   \end{equation}
%   where commas indicate concatenation, and where for
%   $\bmu=(\mu_{1},\ldots,\mu_{p})\in\Z^{p}$ we define
%   \[ k \bmu+j =  (k \mu_1+j,\ldots, k\mu_{p}+j).\]
% \end{definition}

The following result follows directly from the above definitions.
\begin{proposition}
  We have $M=\Theta\left(M^{(0)}, M^{(1)},\ldots M^{(k-1)}\right)$ if and only if
  $\left(M^{(0)}, M^{(1)},\ldots M^{(k-1)}\right)$ is the $k$-modular decomposition of $M$.
  % For $i=0,1,\ldots, k-1$ define the
  % Maya diagrams
  % \[ M_i = \{ m\in \Z \colon km+i \in M\},\quad i=0,1,\ldots, k-1.\]
  % Then,
  % $M_0,\ldots, M_{k-1}$ are the unique Maya diagrams such that
  % \begin{equation}
  %   \label{eq:MMk}
  %   M=[M_0|M_1| \cdots | M_{k-1}]_k.
  % \end{equation}
\end{proposition}

Even though the above operations of interlacing and modular
decomposition apply to general sets, they have a well defined
restriction to Maya diagrams.  Indeed, it is not hard to check that if
$M=\Theta\left(M^{(0)}, M^{(1)},\ldots M^{(k-1)}\right)$ and $M$ is a Maya diagram, then
 $M^{(0)}, M^{(1)},\ldots M^{(k-1)}$ are also Maya diagrams.  Conversely, if the latter are all Maya
diagrams, then so is $M$.  Another important case concerns the
interlacing of finite sets.  The definition \eqref{eq:interlacing}
implies directly that if $\bmu^{(i)} \in \cZ_{p_i},\; i=0,1,\ldots, k-1$
then
\[ \bmu = \Theta\left(\bmu^{(0)}, \ldots , \bmu^{(k-1)}\right) \]
is a finite set of cardinality $p=p_0+\cdots + p_{k-1}$.

Visually, each of the $k$ Maya diagrams is dilated by a factor of $k$,
shifted by one unit with respect to the previous one and superimposed,
so the interlaced Maya diagram incorporates the information from
$M^{(0)}, \ldots M^{(k-1)}$ in $k$ different modular classes. An example can
be seen in Figure~\ref{fig:interlacing}. In other words, the
interlaced Maya diagram is built by copying sequentially a filled or
empty box as determined by each of the $k$ Maya diagrams.

\begin{figure}[h]
\begin{tikzpicture}[scale=0.6]

\draw  (1,3) grid +(11 ,1);

\path [fill,color=black] (0.5,3.5) node {\huge ...} 
++(1,0) circle (5pt) ++(1,0) circle (5pt)  ++(1,0) circle (5pt) 
++(1,0) circle (5pt)
++(2,0) circle (5pt) ++(1,0) circle (5pt)  ++(1,0) circle (5pt) 
++ (4,0) node {\huge ...} +(1,0) node[anchor=west,color=black] { $M_0 = \Xi(0,1,4),\;\qquad\quad\,\,\,\, g_0 = 1$}; 

\draw[line width=1pt] (5,3) -- ++ (0,2);

\foreach \x in {-4,...,6} 	\draw (\x+5.5,4.5)  node {$\x$};

\draw  (1,1) grid +(11 ,1);

\path [fill,color=blue] (0.5,1.5) node {\huge ...} 
++(1,0) circle (5pt) ++(1,0) circle (5pt)  ++(1,0) circle (5pt) 
++(3,0) circle (5pt) ++(1,0) circle (5pt) 
++ (3,0) circle (5pt)  ++(2,0) node {\huge ...} +(1,0) node[anchor=west,color=black] { $M_1 = \Xi(-1,1,3,5,6),\;\quad g_1 = 2$}; 

\draw[line width=1pt] (5,1) -- ++ (0,2);

\draw  (1,-1) grid +(11 ,1);

\path [fill,color=red] (0.5,-0.5) node {\huge ...} 
++(1,0) circle (5pt) ++(1,0) circle (5pt)  ++(1,0) circle (5pt) 
++(1,0) circle (5pt) ++(1,0) circle (5pt)  ++(1,0) circle (5pt) 
++(1,0) circle (5pt)  ++(1,0) circle (5pt) 
++(4,0) node {\huge ...} +(1,0) node[anchor=west,color=black] { $M_2 = \Xi(4),\qquad\qquad\qquad g_2 = 0$}; 

\draw[line width=1pt] (5,-1) -- ++ (0,2);

\draw  (0,-4) grid +(23 ,1);
\foreach \x in {-5,...,17} 	\draw (\x+5.5,2.5-5)  node {$\x$};
\draw[line width=1pt] (5,-4) -- ++ (0,2);

\path [fill,color=black] (2.5,-3.5)   
circle (5pt) ++(6,0) circle (5pt)  
++(3,0) circle (5pt) ++(3,0) circle (5pt) ;

\path [fill,color=blue] 
(0.5,-3.5) circle (5pt) ++(9,0) circle (5pt)  
++(3,0) circle (5pt) ++(9,0) circle (5pt) ;

\path [fill,color=red] 
(1.5,-3.5) circle (5pt) ++(3,0) circle (5pt)  
++(3,0) circle (5pt) ++(3,0) circle (5pt)
++(3,0) circle (5pt) ++(3,0) circle (5pt) ;

\draw (3.5,-5) node[right] {$M= \Theta(M_0,M_1,M_2)=\Xi(0,1,4|-1,1,3,5,6|4) = \Xi(-2,-1,0,2,10,11,12,16,17)$}; 
\end{tikzpicture}

\caption{Interlacing of three Maya diagrams with genus $1,2$ and $0$ with block coordinates and \mbox{$3$-block} coordinates for the interlaced Maya diagram.}\label{fig:interlacing}
\end{figure}

Equipped with these notions of genus and interlacing, we are now ready
to state the main result for the classification of cyclic Maya
diagrams.

% The previous proposition, together with interlacing, is enough to
% prove the main theorem characterizing cyclic Maya diagrams, which we
% now state.

\begin{theorem}
  \label{thm:Mp}
  Let $M=\Theta\left(M^{(0)}, M^{(1)},\ldots M^{(k-1)}\right)$ be the $k$-modular decomposition of a given Maya diagram $M$.  Let $g_i$ be the genus
  of $M^{(i)},\; i=0,1,\ldots, k-1$.  Then, $M$ is $(p,k)$-cyclic where
  \begin{equation}
    \label{eq:pgi}
    p = p_0+p_1+\cdots + p_{k-1},\qquad p_i = 2g_i + 1.
  \end{equation}
\end{theorem}
\begin{proof}
  Let $\bbeta^{(i)}=\Upsilon\left(M^{(i)},M^{(i+1)}\right) \in \cZ_{p_i}$ be the block
  coordinates of $M^{(i)},\; i=0,1,\ldots, k-1$.  Consider the interlacing
  $\bmu = \Theta\left(\bbeta^{(0)}, \ldots , \bbeta^{(k-1)}\right)$.  From
  Proposition \ref{prop:1cyclic} we have that,
  \[ \phi_{\bbeta^{(i)}}\left( M^{(i)}\right)= M^{(i)} +1.\]
  so it follows that
  \begin{align*}
    \phi_{\bmu}(M)&= \phi_{\Theta\left(\bbeta^{(0)}, \ldots , \bbeta^{(k-1)}\right)}\Theta\left( M^{(0)} , \ldots ,
                  M^{(k-1)}\right) \\   
    &= \Theta\left( \phi_{\bbeta^{(0)}}(M^{(0)}) , \ldots ,
                    \phi_{\bbeta^{(k-1)}}(M^{(k-1)})\right) \\
                  &= \Theta\left(M^{(0)}+1 , \ldots , M^{(k-1)} + 1\right) \\
                  &= \Theta\left(M^{(0)}, \ldots , M^{(k-1)}\right) + k \\
                  &= M+k.
  \end{align*}
  Therefore, $M$ is $(p,k)$ cyclic where the value of $p$ agrees with
  \eqref{eq:pgi}.
  % Conversely, suppose that $M\in \cM_g$ is a $(p,k)$-cyclic Maya
  % diagram.  Define the Maya diagrams
  % \[ M_i = \{ m\in \Z \colon km+i \in M\},\quad i=0,1,\ldots, k-1,\]
  % so that \eqref{eq:MMk} holds by construction.  Let $\bmu\in \cZ^p$
  % be a multi-flip such that
  % \[\phi_{\bmu}(M) = M+k \]
  % Let $\Upsilon = \{ \mu_1,\ldots, \mu_p \}$ be the set of all the
  % flip sites. Set
  % \[ \Upsilon_i = \{ m\in \Upsilon \colon k \mid (m- i) \},\quad i=0,1,\ldots, k-1 \]
  % so that
  % \[ \Upsilon = \Upsilon_0 \sqcup \Upsilon_1 \sqcup \cdots \sqcup
  % \Upsilon_{k-1} \]
  % is a partition of the flip sites into modular classes.  For
  % $i=0,1,\ldots, k-1$, let $p_i$ be the cardinality of $\Upsilon_i$
  % and define a multi-flip
  % $\bmu_i = (\mu_{i1} ,\ldots, \mu_{ip_i})\in \Z^{p_i}$ such that
  % \[ \Upsilon_i = \{ k \mu_{ij} + i \colon j=1,\ldots, p_i \}.\] Since
  % \[ M+k = [ M_0+1 | \ldots | M_{k-1}+1]_k,\]
  % it follows that 
  % \[ \phi_{\bmu_i}(M_i) = M_i+1,\quad i=0,1,\ldots, k-1.\]
  % Proposition \ref{prop:1cyclic} requires then that $p_i = 2g_i+1$.
\end{proof}

Theorem~\ref{thm:Mp} sets the way to classify cyclic Maya diagrams for
any given period $p$.
\begin{corollary}\label{cor:k}
  For a fixed period $p\in \N$, there exist $p$-cyclic Maya diagrams with
  shifts $k=p,p-2,\dots,\lfloor p/2\rfloor$, and no other positive shifts are  possible.
\end{corollary}

\begin{remark}
  The highest shift $k=p$ corresponds to the interlacing of $p$ trivial (genus
  0) Maya diagrams.
  % For $n=3$ the $(3,3)$-cyclic Maya diagrams $M$ define the so called Okamoto polynomials $\tau_M$,  \cite{}. The pattern 
\end{remark}

We now introduce a combinatorial system for describing rational
solutions of $p$-cyclic factorization chains.  First, we require a
suitably generalized notion of block coordinates suitable for
describing $p$-cyclic Maya diagrams.
\begin{definition}
  For $p_0,\ldots, p_{k-1} \in \N$ set
  \[ \cZ_{p_0,\ldots, p_{k-1}} := \cZ_{p_0} \times \cdots \times
    \cZ_{p_{k-1}}\subset \Z^{p_0+\cdots +p_{k-1}}. \] Thus, an element
  of $\cZ_{p_0,\ldots, p_{k-1}}$ is a concatenation
  $(\bbeta^{(0)} | \bbeta^{(1)} | \ldots | \bbeta^{(k-1)})$ of $k$
  non-decreasing subsequences,
  $\bbeta^{(i)} \in \cZ_{p_i},\; i=0,1,\ldots, k-1$.  Let
  $\Xi\colon \cZ_{p_0,\ldots, p_{k-1}} \to \cM$ be the mapping with
  action
  \[ \Xi \colon (\bbeta^{(0)} | \bbeta^{(1)} | \ldots |
    \bbeta^{(k-1)})\mapsto \Theta\left(\Xi(\bbeta^{(0)}), \ldots ,
      \Xi(\bbeta^{(k-1)})\right) \] Let $M\in \cM$ be a $(p,k)$ cyclic
  Maya diagram, and let
  \[M=\Theta\left(M^{(0)},\ldots M^{(k-1)}\right)\] be the
  corresponding $k$-modular decomposition.  Let
  $p_i = 2g_i+1,\; i=0,1,\ldots, k-1$ where $g_i$ is the genus of
  $M^{(i)}$ and let $\bbeta^{(i)}\in \hcZ_{2g_i+1}$ be the block
  coordinates of $M^{(i)}\in \cM_{g_i}$.  In light of the fact that
  \[ M = \Theta\left(\Xi(\bbeta^{(0)}), \ldots , \Xi(\bbeta^{(k-1)})\right),\]
  we will refer to the  concatenated sequence
  
  \begin{align*}
    (\bbeta^{(0)} | \bbeta^{(1)} | \ldots |
    \bbeta^{(k-1)}) =  \left( \beta^{(0)}_{0}, \ldots, \beta^{(0)}_{p_0-1} |
    \beta^{(1)}_{0}, \ldots, \beta^{(1)}_{p_1-1} 
    | \ldots | \beta^{(k-1)}_{0},\ldots, \beta^{(k-1)}_{p_{k-1}-1}\right) 
  \end{align*}
  as the $k$-block coordinates of $M$.
\end{definition}

\begin{definition}
  Fix a $k\in \N$. For $m\in \Z$ let $[m]_k\in \{ 0,1,\ldots, k-1 \}$
  denote the residue class of $m$ modulo division by $k$.  For
  $m,n \in \Z$ say that $m \preccurlyeq_k n$ if and only if
  \[ [m]_k < [n]_k,\quad \text{ or } \quad [m]_k = [n]_k\; \text{ and
  } m\leq n.\]
  In this way, the transitive, reflexive relation $\preccurlyeq_k$
  forms a total order on $\Z$.
\end{definition}

\begin{proposition}
  \label{prop:mufrombeta}
  Let $M$ be a $(p,k)$ cyclic Maya diagram.  There exists a unique
  $p$-tuple $\bmu\in \Z^p$ ordered relative to $\preccurlyeq_k$ such
  that
  \begin{equation}
    \label{eq:bmuMk}
    \phi_\bmu(M) = M+k 
  \end{equation}
\end{proposition}
\begin{proof}
  Let
  $(\beta_0,\ldots, \beta_{p-1}) =  (\bbeta^{(0)} | \bbeta^{(1)} | \ldots |   \bbeta^{(k-1)})$ be the $k$-block coordinates of $M$.
  % $M=[M_0| \cdots | M_{k-1}]_k$ be the $k$-modular decomposition
  % of $M$.  Let $(p_0, \ldots, p_{k-1})$ be the $k$-signature of $M$.
  % By definition, $M_i \in \cZ_{p_i}$ for all $i$. Let
  % $\mu_{i0} < \cdots < \mu_{i,p_i-1}$ be the block coordinates of $M_i,\;
  % i=0,1,\ldots, k-1$ arranged in increasing order.   
  % Set
  % \begin{align*}
  %   \bmu_i &= \{ \mu_{i0},\ldots, \mu_{i,p_i-1}\},\quad i=0,1,\ldots,
  %            k-1\\
  %   \bmu &= [\bmu_0| \cdots | \bmu_{k-1}]_k
  % \end{align*}
  % By Proposition \ref{prop:1cyclic} and the proof of Theorem
  % \ref{thm:Mp},
  % \begin{align*}
  %   \phi_{\bmu_i}(M_i) 
  %   &= M_i+1,\quad i=0,1,\ldots, k-1 \\
  %   \phi_{\bmu}(M) 
  %   &= M+k \intertext{where}
  % \end{align*}
  % By the definition of the interlacing operation
  % \eqref{eq:interlacing}, 
  Set \[ \bmu = \Theta\left(\bbeta^{(0)}, \ldots , \bbeta^{(k-1)}\right) \]
  so that \eqref{eq:bmuMk} holds by the proof to Theorem \ref{thm:Mp}.
  The desired enumeration of $\bmu$ is given by
  \[ (k \beta_0,\ldots, k \beta_{p-1}) + (0^{p_0}, 1^{p_1}, \ldots,
  (k-1)^{p_{k-1}}) \]
  where the exponents indicate repetition.  Explicitly,
  $(\mu_0,\ldots, \mu_{p-1})$ is given by
  \[ \left(k\beta^{(0)}_{0},\ldots, k\beta^{(0)}_{p_0-1}, k\beta^{(1)}_{0} + 1 ,\ldots,
  k\beta^{(1)}_{p_1-1} +1 , \ldots , k \beta^{(k-1)}_{0} + k-1, \ldots,
  k\beta^{(k-1)}_{p_{k-1}-1} + k-1 \right) .\]
\end{proof}
\begin{definition}
  In light of \eqref{eq:bmuMk} we will refer to the just defined tuple
  $(\mu_0,\mu_1,\ldots, \mu_{p-1})$ as the $k$-canonical flip sequence
  of $M$ and refer to the tuple $(p_0,p_1,\ldots, p_{k-1})$ as the
  $k$-signature of $M$.
\end{definition}

By Proposition \ref{prop:Mtauw} a rational solution of the
$p$-cyclic dressing chain requires a $(p,k)$ cyclic Maya diagram, and
an additional item data, namely a fixed ordering of the canonical
flip sequence.
We will specify such ordering  as
\[ \bmu_\bpi = (\mu_{\pi_0},\ldots, \mu_{\pi_{p-1}}) \] where
$\bpi=(\pi_0,\ldots, \pi_{p-1})$ is a permutation of
$(0,1,\ldots, p-1)$. With this notation, the chain of Maya diagrams
described in Proposition \ref{prop:Mtauw} is generated as
\begin{equation}\label{eq:picycle}
 M_0 = M,\qquad M_{i+1} =  \phi_{\mu_{\pi_i}}(M_i),\qquad i=0,1,\ldots, p-1.
 \end{equation}
\begin{remark}\label{rem:normalization}
Using a translation it is possible to normalize $M$ so that
$\mu_{0} = 0$.  Using a cyclic permutation and it is possible to
normalize $\bpi$ so that $\pi_p=0$.  The net effect of these two
normalizations is to ensure  that
$M_0,M_1,\ldots, M_{p-1}$ have standard form.
\end{remark}

\begin{remark}
  In order to obtain a full classification of rational solutions, it
  will be necessary to account for degenerate chains which include
  multiple flips at the same site. For this reason, we must allow
  $\bbeta^{(i)}\in \cZ_{p_i}$ to be merely non-decreasing sequences.
\end{remark}

%%%%%%%%%%%%%%%%%%%%%%%%%%%%%%%%%%%%% 
\section{Hermite-type rational solutions}\label{sec:A4}
%%%%%%%%%%%%%%%%%%%%%%%%%%%%%%%%%%%%%
In this section we will put together all the results derived above in
order to describe an effective way of labelling and constructing
Hermite-type rational solutions to the Noumi-Yamada-Painlev\'e system
using cyclic Maya diagrams. As an illustrative example, we describe
rational solutions of the \PIV and \PV systems, because these are
known to be reductions of the $A_2$ and $A_3$ systems, respectively.
We then give examples of rational solutions to the $A_4$ system.

In order to specify a Hermite-type rational solution of a $p$-cyclic
dressing chain, we require three items of data.
\begin{enumerate}
\item We begin by specifying a signature sequence
  $(p_0,\ldots, p_{k-1})$ consisting of odd positive integers that sum
  to $p$. This sequence determines the genus
  $g_i = 2p_i+1,\; i=0,1,\ldots, k-1$ of the $k$ interlaced Maya
  diagrams that give rise to a $(p,k)$-cyclic Maya diagram $M$. The
  possible values of $k$ are given by Corollary~\ref{cor:k}.
\item Once the signature is fixed, we specify an element of
  $\cZ_{p_0,\ldots, p_{k-1}}$; i.e., $k$-block coordinates
  \[(\beta_0,\dots,\beta_{p-1})=(\bbeta^{(0)}|\ldots|
    \bbeta^{(k-1)})\] which determine a $(p,k)$-cyclic Maya diagram
  $M=\Xi(\bbeta^{(0)}|\ldots| \bbeta^{(k-1)})$, and a canonical flip
  sequence $\bmu=(\mu_0,\dots,\mu_{p-1})$ as per Proposition
  \ref{prop:mufrombeta}.
\item Once the $k$-block coordinates and canonical flip sequence
  $\bmu$ are fixed, we specify a permutation
  $(\pi_0,\ldots, \pi_{p-1})$ of $(0,1,\ldots, p-1)$ that determines
  the actual flip sequence $\bmu_\bpi$, i.e. the order in which the
  flips in the canonical flip sequence are applied to build a
  $p$-cycle of Maya diagrams.
\item With the above data, we apply Proposition \ref{prop:Mtauw} with
  $M$ and $\bmu_\bpi$ to construct the rational solution.
\end{enumerate}

For any signature of a Maya $p$-cycle, we need to specify the $p$
integers in the canonical flip sequence, but following
Remark~\ref{rem:normalization}, we can get rid of translation
invariance by imposing $\mu_0=\beta^{(0)}_0=0$, leaving only $p-1$
free integers.  The remaining number of degrees of freedom is $p-1$,
which coincides with the number of generators of the symmetry group
$A^{(1)}_{p-1}$. This is a strong indication that the class described
above captures a generic orbit of a seed solution under the action of
the symmetry group.  Moreover, it is sometimes advantageous to consider
only permutations such that $\pi_p=0$ in order to remove the
invariance under cyclic permutations.

\subsection{Painlev\'e IV}

As was mentioned in the introduction, the $A_2$ Noumi-Yamada system is
equivalent the \PIV equation
\cite{adler1994nonlinear,veselov1993dressing}.  The reduction of
\eqref{eq:P4system} to the \PIV equation \eqref{eq:P4scalar} is
accomplished via the following substitutions
% \begin{gather}
%   \label{eq:f123W}
%   w_1(z) = -W(z) - \frac{\Delta}{2} z,\quad w_{2,3}(z) =
%   \frac{W(x)}{2} \pm
%   \frac{ W'(x)-a_2}{2W(x)},\quad \Delta = -(a_1+a_2+a_3),,\\
%   \nonumber W(x) = 2^{-\frac12} \Delta^{\frac12}\, y(t),\;
%   x=2^{\frac12} \Delta^{-\frac12}\, t,\qquad a =
%   \Delta^{-1}(a_3-a_1),\; b = -2\Delta^{-2}a_2^2.
% \end{gather}
\begin{gather}
  \label{eq:f123y}
  \sqrt{2}f_0(x)=y(t),\quad -\sqrt{8}f_{1,2}(x) =y(t)+2t \pm
  \frac{y'(t) - \sqrt{-2b}}{y(t)} \\
  \nonumber x=-\sqrt{2}t,\qquad a = \alpha_2-\alpha_1,\; b =
  -2\alpha_0^2.
\end{gather}

It is known \cite{clarkson2003fourth} that every rational solution of
\PIV can be described in terms of either generalized Hermite (GH) or
Okamoto (O) polynomials, both of which may given as a Wronskian of
classical Hermite polynomials.  We now exhibit these solutions using
the framework described in the preceding section.

The 3-cyclic Maya diagrams fall into exactly one of two classes:
\[
  \begin{array}{cccc}
    k & (p_0,\ldots, p_{k-1})& (\beta_0,\beta_1, \beta_2)
    &  (\mu_0,\mu_1, \mu_{2})     \\
    1 &(3)& (0,n_1,n_1+n_2)&(0,n_1,n_1+n_2) \\
    3 &(1,1,1) & (0|n_1|n_2)& (0,3 n_1+1,3 n_2+2)
\end{array}
\]
The corresponding Maya diagrams are
\begin{align*}
  M_{\MGH}(n_1,n_2)
  &= \Xi(0,n_1,n_1+n_2) = (-\infty,0) \cup   [n_1,n_1+n_2) \\
  M_{\MO}(n_1,n_2)
  &= \Xi(0|n_1|n_2) \\
  &=\Theta((-\infty,0),(-\infty,n_1), (-\infty,n_2))\\
  &= (-\infty,0) \cup \{ 3j+1 \colon 0\leq j<
    n_1\} \cup \{ 3j+2 \colon 0 \leq j \leq n_2 \}
\end{align*}
The generalized Hermite and Okamoto polynomials, denoted below by
$\tau_{\text{GH}(n_1,n_2)},\; n_1,n_2\geq 0$ and
$\tau_{\text{O}(n_1,n_2)},\;n_1,n_2\geq 0$, respectively, are
two-parameter families of Hermite Wronskians that correspond to the
above diagrams.  For example (see \eqref{eq:tauWr} for definition):
\begin{align*}
  \tau_{\MGH(3,5)} &= \tau(3,4,5,6,7),\\
  \tau_{\MO(3,2)} &= \tau(1,2,4,5,7)
\end{align*}

Having chosen one of the above polynomials as $\tau_0$ there are
$6=3!$ distinct rational solutions of the $A_2$ system
\eqref{eq:P4system} corresponding the possible permutations of the
canonical flip sites $(\mu_0,\mu_1,\mu_2)$.  The translations of the
block coordinates $(\beta_0,\beta_1,\beta_2)$ engendered by these
permutations is enumerated in the table below.

\[
  \begin{array}{ccccccc}
    & (012) & (102) & (021) & (210) & (120) & (201) \\
    \tau_0&000&000&000&000&000&000\\
    \tau_1&100&010&100&001&010&001\\
    \tau_2&110&110&101&011&011&101\\
    \tau_3& 111& 111& 111& 111& 111& 111\\
  \end{array}
\]

%   The corresponding $\tau$-functions indices are shown below
% \[
%   \begin{array}{ccccccc}
%     & (012) & (102) & (021) & (210) & (120) & (201) \\
%     \tau_0& (n_1,n_2)&(n_1,n_2)&(n_1,n_2)&(n_1,n_2)&(n_1,n_2)&(n_1,n_2)\\
%     \tau_1&(n_1-1,n_2-1)&(n_1+1,n_2-1)&(n_1-1,n_2-1)&(n_1,n_2+1)&(n_1+1,n_2-1)&(n_1,n_2+1)\\    
%     \tau_2&(n_1, n_2-1)&(n_1,n_2-1)& (n_1-1,n_2)&(n_1+1,n_2)&(n_1+1,n_2)&(n_1-1,n_2)\\
%     \tau_3& (n_1, n_2)& (n_1, n_2)&(n_1, n_2)& (n_1, n_2)&(n_1,n_2)&(n_1,n_2)
%   \end{array}
% \]

\begin{example}
  The Maya diagram, $M_{\MGH}(3,5)$ has the canonical flip sequence
  $\mu=(0,3,8)$. Applying the permutation $\pi=(2,1,0)$ yields the
  flip sequence $\mu_\pi=(8,3,0)$ and the following sequence of
  polynomials and block coordinates:
  \begin{align*}
  \tau_0 &= \tau(3,4,5,6,7),& (0,3,8)\\
  \tau_1 &= \tau(3,4,5,6,7,8),& (0,3,9)\\
  \tau_2 &= \tau(4,5,6,7,8), & (0,4,9)\\
  \tau_3 &=  \tau(0,4,5,6,7,8) \propto \tau(3,4,5,6,7), & (1,4,9)
  \end{align*}
\begin{figure}[ht]
  \centering
\begin{tikzpicture}[scale=0.5]

  \path [fill] (0.5,6.5)
  circle (5pt) ++ (4,0)
  circle (5pt) ++ (1,0)
  circle (5pt) ++ (1,0)
  circle (5pt) ++ (1,0)
  circle (5pt) ++ (1,0)
  circle (5pt) ++ (3,0)
  node[anchor=west]  {$M_0=\Xi(0,3,8)  = M_{\MGH}(3,5)$};

  \path [fill] (0.5,5.5)
  circle (5pt) ++ (4,0)
  circle (5pt) ++ (1,0)
  circle (5pt) ++ (1,0)
  circle (5pt) ++ (1,0)
  circle (5pt) ++ (1,0)
  circle (5pt) ++ (1,0)
  circle (5pt) ++ (2,0)
  node[anchor=west]  {$M_1=\Xi(0,3,9) = M_{\MGH}(3,6)$};

  \path [fill] (0.5,4.5)
  circle (5pt) ++ (5,0)
  circle (5pt) ++ (1,0)
  circle (5pt) ++ (1,0)
  circle (5pt) ++ (1,0)
  circle (5pt) ++ (1,0)
  circle (5pt) ++ (2,0)
  node[anchor=west]  {$M_2=\Xi(0,4,9) = M_{\MGH}(4,5)$};

  \path [fill] (0.5,3.5)
  circle (5pt) ++ (1,0)
  circle (5pt) ++ (4,0)
  circle (5pt) ++ (1,0)
  circle (5pt) ++ (1,0)
  circle (5pt) ++ (1,0)
  circle (5pt) ++ (1,0)
  circle (5pt) ++ (2,0)
  node[anchor=west]  {$M_3=\Xi(1,4,9)=M_{\MGH}(3,5)+1$};

  \draw  (0,3) grid +(11 ,4);
  \draw[line width=2pt] (1,3) -- ++ (0,4);
  \draw[line width=2pt] (10,3) -- ++ (0,4);

  \foreach \x in {-1,...,9} \draw (\x+1.5,2.5)  node {$\x$};

\end{tikzpicture}
  
\caption{A cycle generated by the $(3,1)$-cyclic  Maya diagram
  $M_{\MGH}(3,5)$ and  permutation $\pi=(210)$. }
  \label{fig:31cyclic}
\end{figure}

\noindent
Hence, by \eqref{eq:HM2w}
\[ (a_0,a_1,a_2) = (10,6,-18),\qquad (\sigma_0,\sigma_1,\sigma_2) =
  (-1,1,-1)\] Applying \eqref{eq:witaui} gives the following rational
solution of the 3-cyclic dressing chain \eqref{eq:wachain}:
\[
  w_0 = z + \frac{\tau_1'}{\tau_1} - \frac{\tau_{0}'}{\tau_{0}},\quad
  w_1 = -z + \frac{\tau_2'}{\tau_2} - \frac{\tau_{1}'}{\tau_{1}}\quad
  w_2 = z + \frac{\tau_0'}{\tau_0} - \frac{\tau_{2}'}{\tau_{2}},\qquad
  w_i=w_i(z),\; \tau_i = \tau_i(z).
\]
Applying \eqref{eq:ffromw} and \eqref{eq:f123y} gives the following
rational solution of \PIV:
\[ y(t) = \frac{d}{dt} \log \frac{\tau_0(t)}{\tau_2(t)},\quad a= \frac12(a_1-a_2)=12,\;
  b=-\frac{a_0^2}{2} = -50 .\]
\end{example}

\begin{example}
  The Maya diagram $M_{\MO}(3,2)$ has the canonical flip sequence
  \[ \mu=\Theta(0|3|2)=(0,10,8).\] Using the permutation $(2,0,1)$ by
  way of example, we generate the following sequence of polynomials
  and block coordinates:
\begin{align*}
  \tau_0 &= \tau(1,2,4,5,7), & (0|3|2)\\
  \tau_1 &= \tau(1,2,4,5,7,8), & (0|3|3)\\
  \tau_2 &= \tau(0,1,2,4,5,7,8), & (1|3|3)\\
  \tau_3 &= \tau(0,1,2,4,5,7,8,10)\propto
           \tau(1,2,4,5,7) & (1|4|3)\\
\end{align*}
\begin{figure}[ht]
  \centering
\begin{tikzpicture}[scale=0.5]

  \path [fill] (0.5,6.5)  circle (5pt);
  \path [fill,color=blue] (1.5,6.5)
  circle (5pt) ++ (3,0)
  circle (5pt) ++ (3,0)
  circle (5pt) ++ (3,0)
  circle (5pt);
  \path [fill,color=red] (2.5,6.5)
  circle (5pt) ++ (3,0)
  circle (5pt) ++ (3,0)
  circle (5pt);

  \path [fill] (0.5,5.5)  circle (5pt);
  \path [fill,color=blue] (1.5,5.5)
  circle (5pt) ++ (3,0)
  circle (5pt) ++ (3,0)
  circle (5pt) ++ (3,0)
  circle (5pt);
  \path [fill,color=red] (2.5,5.5)
  circle (5pt) ++ (3,0)
  circle (5pt) ++ (3,0)
  circle (5pt) ++ (3,0)
  circle (5pt);

  \path [fill] (0.5,4.5)
  circle (5pt) ++ (3,0)
  circle (5pt);
  \path [fill,color=blue] (1.5,4.5)
  circle (5pt) ++ (3,0)
  circle (5pt) ++ (3,0)
  circle (5pt) ++ (3,0)
  circle (5pt);
  \path [fill,color=red] (2.5,4.5)
  circle (5pt) ++ (3,0)
  circle (5pt) ++ (3,0)
  circle (5pt) ++ (3,0)
  circle (5pt);

  \path [fill] (0.5,3.5)  circle (5pt) ++ (3,0)   circle (5pt) ;
  \path [fill,color=blue] (1.5,3.5)
  circle (5pt) ++ (3,0)
  circle (5pt) ++ (3,0)
  circle (5pt) ++ (3,0)
  circle (5pt) ++ (3,0)
  circle (5pt);
  \path [fill,color=red] (2.5,3.5)
  circle (5pt) ++ (3,0)
  circle (5pt) ++ (3,0)
  circle (5pt) ++ (3,0)
  circle (5pt);
  
  \draw (16.5,6.5)  node[anchor=west]  {$M_0=\Xi(0|3|2)= M_{\MO}(3,2)$};
  \draw (16.5,5.5)  node[anchor=west]  {$M_1=\Xi(0|3|3)=M_{\MO}(3,3)$};
  \draw (16.5,4.5)  node[anchor=west] {$M_2=\Xi(1|3|3)=M_{\MO}(2,2)+3$};
  \draw (16.5,3.5)  node[anchor=west]  {$M_3=\Xi(1|4|3)=M_{\MO}(3,2)+3$};

  \draw  (0,3) grid +(16 ,4);
  \draw[line width=2pt] (3,3) -- ++ (0,4);
  \draw[line width=2pt] (15,3) -- ++ (0,4);

  \foreach \x in {-3,...,12} \draw (\x+3.5,2.5)  node {$\x$};
\end{tikzpicture}
  
\caption{A cycle generated by the $(3,3)$-cyclic  Maya diagram
  $M_{\MO}(3,2)$ and  permutation $\pi=(201)$. }
  \label{fig:31cyclic}
\end{figure}

\noindent
Hence, by \eqref{eq:HM2w}
\[ (a_0,a_1,a_2) = (16,-20,-2),\qquad (\sigma_0,\sigma_1,\sigma_2) =
  (-1,-1,-1)\] Applying \eqref{eq:witaui} gives the following rational
solution of the 3-cyclic dressing chain \eqref{eq:wachain}:
\[
  w_0 = \frac{\tau_1'}{\tau_1} - \frac{\tau_{0}'}{\tau_{0}},\quad w_1
  = \frac{\tau_2'}{\tau_2} - \frac{\tau_{1}'}{\tau_{1}}\quad w_2 =
  \frac{\tau_0'}{\tau_0} - \frac{\tau_{2}'}{\tau_{2}}.
\]

Applying \eqref{eq:ffromw} and \eqref{eq:f123y} gives the following
rational solution of \PIV:
\[ y(t) = -\frac{2}{3}t + \frac{d}{dt} \log
  \frac{\tau_2(Kt)}{\tau_0(Kt)}, \; K=\frac{1}{\sqrt{3}},\quad a=
  \frac16(a_1-a_2)=-3,\; b=-\frac{a_0^2}{18} = -\frac{128}{9} .\]
\end{example}

For each of the above classes of solutions, observe that the
permutations $(\pi_0,\pi_1,\pi_2)$ and $(\pi_1,\pi_0,\pi_2)$, while
producing distinct solutions of the $A_2$ system, give the same
solution of \PIV.  This means that there are 3 distinct $\MGH$ classes
and 3 distinct $\MO$ classes of rational solution of \PIV.  These are
enumerated below.

\begin{align*}
  y &= \frac{d}{dt} \log
      \frac{\tau_{\MGH(n_1,n_2)}(t)}{\tau_{\MGH(n_1,n_2+1)}(t)},
  && a =-(1+n_1+2n_2),\; b = -2 n_1^2,\\
  y &= \frac{d}{dt} \log
      \frac{\tau_{\MGH(n_1,n_2)}(t)}{\tau_{\MGH(n_1-1,n_2)}(t)},
  &&   a=2n_1+n_2-1,\; b=-2n_2^2\\
  y &= -2t+\frac{d}{dt} \log
      \frac{\tau_{\MGH(n_1,n_2)}(t)}{\tau_{\MGH(n_1+1,n_2-1)}(t)},
  && a=n_2-n_1-1,\;  b = -2(n_1+n_2)^2,\\
  y &= -\frac23t+\frac{d}{dt} \log
      \frac{\tau_{\MO(n_1,n_2)}(z)}{\tau_{\MO(n_1-1,n_2-1)}(z)},\; z = \frac{t}{\sqrt{3}}
      ,&&   a= n_1+n_2,\;  b = -\frac29
          (1-3n_1+3n_2)^2  \\
  y &=  -\frac23t+\frac{d}{dt} \log
      \frac{ \tau_{\MO(n_1,n_2)}(z)}{ \tau_{\MO(n_1+1,n_2)}(z)},
  &&   a= -1-2n_1+n_2,\;  b = -\frac29 (2+3n_2)^2  \\
  y &=  -\frac23t+\frac{d}{dt}
      \frac{ \tau_{\MO(n_1,n_2)}(z)}{ \tau_{\MO(n_1,n_2+1)}(z)},
  &&   a= -2-2n_2+n_1,\;  b = -\frac29
     (1+3n_1)^2  
\end{align*}

\subsection{Painlev\'e V}
\label{sect:PV}

The fifth Painlev\'e equation is a second-order scalar non-autonomous,
non-linear differential equation, usually given as
\begin{equation}
  \label{eq:P5scalar}
  y'' = (y')^2\left(\frac{1}{2y}+ \frac{1}{y-1}\right) -\frac{y'}{t} +
  \frac{(y-1)^2}{t^2} \left(ay+\frac{b}{y}\right)+\frac{cy}{t} +
  \frac{dy(y+1)}{y-1},\quad y=y(t),
\end{equation}
where $a,b,c,d$ are complex-valued parameters. An equivalent form is
\begin{align*}
  \phi &= \alpha_0(y-1) + \alpha_2 \left(1-\frac1y\right) - t\,
         \frac{y'}{y},\quad\phi=\phi(t) \\
  \phi' &= \frac{1}{2t} \left(\frac{\phi(y(2-\phi)-\phi-2)}{y-1} +
          2\phi(\alpha_0+\alpha_2-1) - 2 c t\right) - d t\,
          \frac{y+1}{y-1}\intertext{where}
  a &= \frac{\alpha_0^2}{2},\; b= -\frac{\alpha_2^2}{2}          
\end{align*}

The $A_3$ Noumi-Yamada system, the specialization of \eqref{eq:Aodd}
with $n=2$, has the form
\begin{equation}
  \label{eq:A3system}
  \begin{aligned}
    x f_0'  &= 2 f_0 f_2 (f_1-f_3) + (1-2\alpha_2) f_0 + 2 \alpha_0 f_2 \\
    x f_1'  &= 2 f_1 f_3 (f_2-f_0) + (1-2\alpha_3) f_1 + 2 \alpha_1 f_3 \\
    x f_2'  &= 2 f_0 f_2 (f_3-f_1) + (1-2\alpha_0) f_2 + 2 \alpha_2 f_0 \\
    x f_3'  &= 2 f_1 f_3 (f_0-f_2) + (1-2\alpha_1) f_3 + 2 \alpha_3 f_1 
  \end{aligned}
\end{equation}
where $f_i=f_i(x),\; i=0,1,2,3$, and which is
subject to normalizations
\begin{equation}
  \label{eq:A3norm}
  f_0 + f_2 = f_1+ f_3 = \frac{x}{2} ,\qquad \alpha_0 + \alpha_1 +
  \alpha_2 + \alpha_3 = 1.
\end{equation}
The transformation of \eqref{eq:A3system} to \eqref{eq:P5scalar} is
given by the following relations:
\begin{equation}
  \label{eq:yP5fromf}
  \begin{aligned}
    y &= -\frac{f_2}{f_0},\qquad \phi=\frac12 x (f_1-f_3) ,\\\
    y&=y(t),\;\phi=\phi(t),\quad f_{i} = f_{i}(x),\quad t=
    \frac{x^2}{\Delta}\\
    a &= \frac{\alpha_0^2}{2},\; b= -\frac{\alpha_2^2}{2},\;
    c= \frac{\Delta}{4}(\alpha_3-\alpha_1) ,\; d=a
    -\frac{\Delta^2}{32}.
  \end{aligned}
\end{equation}

The 4-cyclic Maya diagrams fall into exactly one of three classes:
\[
  \begin{array}{cccc}
    k & (p_0,\ldots, p_{k-1})& (\beta_0,\beta_1, \beta_2,\beta_3)
    &  (\mu_0,\mu_1, \mu_{2},\mu_3)     \\
    2 &(3,1)& (0,n_1,n_1+n_2|n_3) & (0,2n_1, 2n_1+2n_2,2n_3+1) \\
    2 &(1,3)& (0|n_1,n_1+n_2,n_1+n_2+n_3)
    & (0,2n_1+1,2n_1+2n_2+1,2n_1+2n_2+2n_3+1)    \\ 
    3 &(1,1,1,1) & (0|n_1|n_2|n_3)
      &(0,2n_1+1, 2n_2+2,2n_2+3)
\end{array}
\]
We will denote the corresponding Maya diagrams as
\begin{align*}
  M(n_1,n_2|n_3)  &= \Xi(0,n_1,n_1+n_2|n_3)\\
  M(|n_1,n_2,n_3)&= \Xi(0|n_1,n_1+n_2,n_1+n_2+n_3)\\
  M(|n_1|n_2|n_3)&= \Xi(0|n_1|n_2|n_3)
\end{align*}
For example (see \eqref{eq:tauWr} for definition):
\begin{align*}
  \tau_{M(3,1|2)} &= \tau(1,3,6),\\
  \tau_{M(|3,1,2)} &= \tau(1,3,5,9,11),\\
  \tau_{M(|3|1|2)} &= \tau(1,2,3,5,7,9)
\end{align*}

Having chosen one of the above polynomials as $\tau_0$, there are
$24=4!$ distinct rational solutions of the $A_3$ system
\eqref{eq:P4system} corresponding the possible permutations of the
canonical flip sites $(\mu_0,\mu_1,\mu_2,\mu_4)$.  However, the
projection from the set of $A_3$ solutions to the set of \PV solutions
is not one-to-one.  The action of the permutation group $\fS_4$ on the
set of \PV solutions has non-trivial isotropy corresponding to the
Klein 4-group: $(0,1,2,3), (1,0,2,3), (0,1,3,2), (1,0,3,2)$.
Therefore each of the above $\tau$-functions generates $6=24/4$
distinct rational solutions of \PV.

\begin{example}\label{ex:42}
  We exhibit a $(4,2)$-cyclic Maya diagram in the signature class
  $(3,1)$ by taking
  \[M_0=M(3,1|2)=\Xi(0,3,4|2),\] depicted in the first row of
  Figure~\ref{fig:42cyclic}.  The canonical flip sequence is
  $\bmu=(0,6,8,{\color{blue}5})$.  By way of example, we choose the permutation
  $(0132)$, which gives the chain of Maya diagrams shown in Figure
  \ref{fig:42cyclic} and the following $\tau$ functions:
  \begin{align*}
    \tau_0 &= \tau(1,3,6) \propto z(8z^6-12z^4-6z^2-3)\\
    \tau_1 &= \tau(0,1,3,6)\propto 4z^4-4z^2-1 \\
    \tau_2 &= \tau(0,1,3)\propto z\\
    \tau_3 &= \tau(0,1,3,5)\propto z^3\\
    \tau_4 &= \tau(0,1,3,5,8)\propto \tau(1,3,6)
  \end{align*}
  By \eqref{eq:HM2w},
  \[ (a_0,\ldots, a_3) = (-12,2,-6,12),\qquad (\sigma_0,\ldots,
    \sigma_3) = (-1,1,-1,-1) \] Hence, by \eqref{eq:witaui} and
  \eqref{eq:ffromw},
  \[
    \begin{aligned}
      f_0 &= \frac{6x^5-24x^3-24x}{x^6-6x^4-12x^2-24}\\
      f_1 &= -\frac{3}{x}+ \frac{4x^3-8x}{x^4-4x^2-4}
    \end{aligned}\] with $f_2,f_3$ given by \eqref{eq:A3norm}. By
  \eqref{eq:yP5fromf} the corresponding rational solution of \PV
  \eqref{eq:P5scalar} is given by
  \[
    \begin{aligned}
      y &= \frac76-\frac{t}{3}+ \frac{4t+2}{12t^2-12t-3}\\
      a&= \frac92,\; b= -\frac98,\; c= -\frac52,\; d= -\frac12
    \end{aligned}
    \]

  \begin{figure}[ht]
    \centering
    \begin{tikzpicture}[scale=0.5]
      \path [fill,blue] (0.5,6.5)
      circle (5pt)  ++(2,0)
      circle (5pt)  ++(2,0)
      circle (5pt);
      \path [fill,black] (-0.5,6.5)
      circle (5pt)  ++(8,0)
      circle (5pt)  ++(4,0)
      node[anchor=west] {$M_0=\Xi\,(0,3,4|2)=M(3,1|2)$};
      
      \path [fill,blue] (0.5,5.5)
      circle (5pt)  ++(2,0)
      circle (5pt)  ++(2,0)
      circle (5pt);
      \path [fill,black] (-0.5,5.5)
      circle (5pt)  ++(2,0)
      circle (5pt)  ++(6,0)
      circle (5pt)  ++(4,0)
      node[anchor=west] {$M_1=\Xi\,(1,3,4|2)=M(2,1|1)+2$};
      
      \path [fill,blue] (0.5,4.5)
      circle (5pt)  ++(2,0)
      circle (5pt)  ++(2,0)
      circle (5pt);
      \path [fill,black] (-0.5,4.5)
      circle (5pt)  ++(2,0)
      circle (5pt)  ++(10,0)
      node[anchor=west] {$M_2=\Xi\,(1,4,4|2)=M(3,0|1)+2$};
      
      \path [fill,blue] (0.5,3.5)
      circle (5pt)  ++(2,0)
      circle (5pt)  ++(2,0)
      circle (5pt)  ++(2,0)
      circle (5pt);
      \path [fill,black] (-0.5,3.5)
      circle (5pt)  ++(2,0)
      circle (5pt)  ++(10,0)
      node[anchor=west] {$M_2=\Xi\,(1,4,4|3)=M(3,0|2)+2$};
      
      \path [fill,blue] (0.5,2.5)
      circle (5pt)  ++(2,0)
      circle (5pt)  ++(2,0)
      circle (5pt)  ++(2,0)
      circle (5pt);
      \path [fill,black] (-0.5,2.5)
      circle (5pt)  ++(2,0)
      circle (5pt)  ++(8,0)
      circle (5pt)  ++(2,0)
      node[anchor=west] {$M_2=\Xi\,(1,4,5|3)=M(3,1|2)+2$};

      \draw  (-1,2) grid +(12 ,5);
      \draw[line width=2pt] (1,2) -- ++ (0,5);
      \draw[line width=2pt] (10,2) -- ++ (0,5);
      
      \foreach \x in {0,...,9} \draw (\x+1.5,1.5)  node {$\x$};
    \end{tikzpicture}
    \caption{A Maya $4$-cycle with shift $k=2$ for the choice
      $(n_1,n_2|n_3)=(3,1|2)$ and permutation $\bpi=(0,1,3,2)$. }
    \label{fig:42cyclic}
  \end{figure}
\end{example}

\begin{example}\label{ex:42b}
  We exhibit a $(4,2)$-cyclic Maya diagram in the signature class
  $(1,3)$ by taking
  \[M_0=M(|3,1,2)=\Xi(0|3,4,6),\] depicted in the first row of
  Figure~\ref{fig:42bcyclic}.  The canonical flip sequence is
  $\bmu=(0,{\color{blue} 7,9,13})$.  By way of example, we choose the
  permutation  
  $(0132)$, which gives the chain of Maya diagrams shown in Figure
  \ref{fig:42bcyclic} and the following $\tau$ functions:
  \begin{align*}
    \tau_0 &= \tau(1,3,5,9,11) \propto z^{15}(4z^4-36z^2+99)\\
    \tau_1 &= \tau(0,1,3,5,9,11)\propto z^{10}(4z^4-28z^2+63) \\
    \tau_2 &= \tau(0,1,3,5,7,9,11)\propto z^{15}\\
    \tau_3 &= \tau(0,1,3,5,7,9,11,13)\propto z^{21}\\
    \tau_4 &= \tau(0,1,3,5,7,11,13)\propto \tau(1,3,5,9,11)
  \end{align*}
  By \eqref{eq:HM2w},
  \[ (a_0,\ldots, a_3) = (-12,2,-6,12),\qquad (\sigma_0,\ldots,
    \sigma_3) = (-1,-1,-1,1) \] Hence, by \eqref{eq:witaui} and
  \eqref{eq:ffromw},
  \[
    \begin{aligned}
      f_0 &= \frac{x}{2} + \frac{4x^3-72x^2}{x^4-36x^2+396}\\
      f_1 &= -\frac{11}{x} +\frac{x}{2}+ \frac{4x^3-56x}{x^4-28x^2+252}
    \end{aligned}\] with $f_2,f_3$ given by \eqref{eq:A3norm}. By
  \eqref{eq:yP5fromf} the corresponding rational solution of \PV
  \eqref{eq:P5scalar} is given by
  \[
    y = \frac{8t-36}{4t^2-28t+63},\qquad a= \frac{49}{8},\; b= -2,\;
    c= -\frac{13}{2},\; d= -\frac12
    \]

  \begin{figure}[ht]
    \centering
    \begin{tikzpicture}[scale=0.5]
      \path [fill,blue] (0.5,6.5)
      circle (5pt)  ++(2,0)
      circle (5pt)  ++(2,0)
      circle (5pt)  ++(2,0)
      circle (5pt)  ++(4,0)
      circle (5pt)  ++(2,0)
      circle (5pt);
      \path [fill,black] (-0.5,6.5)
      circle (5pt)   ++(18,0)
      node[anchor=west] {$M_0=\Xi\,(0|3,4,6)=M(|3,1,2)$};
      
      \path [fill,blue] (0.5,5.5)
      circle (5pt)  ++(2,0)
      circle (5pt)  ++(2,0)
      circle (5pt)  ++(2,0)
      circle (5pt)  ++(4,0)
      circle (5pt)  ++(2,0)
      circle (5pt);
      \path [fill,black] (-0.5,5.5)
      circle (5pt)  ++(2,0)
      circle (5pt)   ++(16,0)
      node[anchor=west] {$M_1=\Xi\,(1|3,4,6)=M(|2,1,2)+2$};
      
      \path [fill,blue] (0.5,4.5)
      circle (5pt)  ++(2,0)
      circle (5pt)  ++(2,0)
      circle (5pt)  ++(2,0)
      circle (5pt)  ++(2,0)
      circle (5pt)  ++(2,0)
      circle (5pt)  ++(2,0)
      circle (5pt);
      \path [fill,black] (-0.5,4.5)
      circle (5pt)  ++(2,0)
      circle (5pt)   ++(16,0)
      node[anchor=west] {$M_2=\Xi\,(1|4,4,6)=M(|3,0,2)+2$};
      
      \path [fill,blue] (0.5,3.5)
      circle (5pt)  ++(2,0)
      circle (5pt)  ++(2,0)
      circle (5pt)  ++(2,0)
      circle (5pt)  ++(2,0)
      circle (5pt)  ++(2,0)
      circle (5pt)  ++(2,0)
      circle (5pt)  ++(2,0)
      circle (5pt);
      \path [fill,black] (-0.5,3.5)
      circle (5pt)  ++(2,0)
      circle (5pt)   ++(16,0)
      node[anchor=west] {$M_3=\Xi\,(1|4,4,7)=M(|3,0,3)+2$};
      
      \path [fill,blue] (0.5,2.5)
      circle (5pt)  ++(2,0)
      circle (5pt)  ++(2,0)
      circle (5pt)  ++(2,0)
      circle (5pt)  ++(2,0)
      circle (5pt)  ++(4,0)
      circle (5pt)  ++(2,0)
      circle (5pt);
      \path [fill,black] (-0.5,2.5)
      circle (5pt)  ++(2,0)
      circle (5pt)   ++(16,0)
      node[anchor=west] {$M_4=\Xi\,(1|4,5,7)=M(|3,1,2)+2$};

      \draw  (-1,2) grid +(17 ,5);
      \draw[line width=2pt] (1,2) -- ++ (0,5);
      \draw[line width=2pt] (15,2) -- ++ (0,5);
      
      \foreach \x in {0,...,14} \draw (\x+1.5,1.5)  node {$\x$};
    \end{tikzpicture}
    \caption{A Maya $4$-cycle with shift $k=2$ for the choice
      $(|n_1,n_2,n_3)=(|3,1,2)$ and permutation $\bpi=(0,1,3,2)$. }
    \label{fig:42bcyclic}
  \end{figure}
\end{example}

\begin{example}\label{ex:44}
  We exhibit a $(4,4)$-cyclic Maya diagram in the signature class
  $(1,1,1,1)$ by taking
  \[M_0=M(|3|1|2)=\Xi(0|3|1|2),\] depicted in the first row of
  Figure~\ref{fig:44cyclic}.  The canonical flip sequence is
  $\bmu=(0,{\color{blue}13},{\color{red}6},{\color{green}11})$.  By
  way of example, we choose the permutation 
  $(0132)$, which gives the chain of Maya diagrams shown in Figure
  \ref{fig:44cyclic} and the following $\tau$ functions:
  \begin{align*}
    \tau_0 &= \tau(1,2,3,5,7,9) \propto z^{10}(2z^2+9)\\
    \tau_1 &= \tau(0,1,2,3,5,7,9)\propto z^{6} \\
    \tau_2 &= \tau(0,1,2,3,5,7,9,13)\propto z^{10}(2z^2-9)\\
    \tau_3 &= \tau(0,1,2,3,5,7,9,11,13)\propto z^{15}\\
    \tau_4 &= \tau(0,1,2,3,5,6,7,9,11,13)\propto \tau(1,2,3,5,7,9)
  \end{align*}
  By \eqref{eq:HM2w},
  \[ (a_0,\ldots, a_3) = (-26,4,10,4),\qquad (\sigma_0,\ldots,
    \sigma_3) = (-1,-1,-1,-1) \] Hence, by \eqref{eq:witaui} and
  \eqref{eq:ffromw},
  \[
    \begin{aligned}
      f_0 &= -\frac{1}{x-6} -\frac{1}{x+6}+\frac{x}{4}+ \frac{2x}{x^2+36}\\
      f_1 &= -\frac{9}{x} +\frac{x}{4},
    \end{aligned}\] with $f_2,f_3$ given by \eqref{eq:A3norm}. By
  \eqref{eq:yP5fromf} the corresponding rational solution of \PV
  \eqref{eq:P5scalar} is given by
  \[
    y = -1 -\frac{72}{4t^2-117},\qquad a= \frac{169}{32},\; b= -\frac{25}{32},\;
    c=0,\; d= -2
    \]

  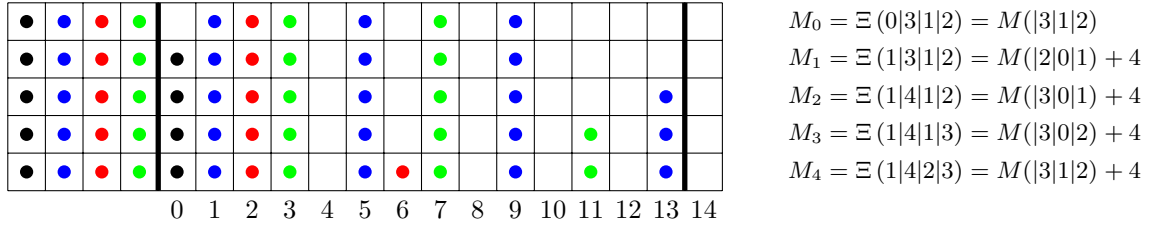
\begin{figure}[ht]
    \centering
    \begin{tikzpicture}[scale=0.5]
      \path [fill,black] (-2.5,6.5)
      circle (5pt)   ++(20,0)
      node[anchor=west] {\small$M_0=\Xi\,(0|3|1|2)=M(|3|1|2)$};
      \path [fill,blue] (-1.5,6.5)
      circle (5pt)  ++(4,0)
      circle (5pt)  ++(4,0)
      circle (5pt)  ++(4,0)
      circle (5pt);
      \path [fill,red] (-0.5,6.5)
      circle (5pt)  ++(4,0)
      circle (5pt);
      \path [fill,green] (0.5,6.5)
      circle (5pt)  ++(4,0)
      circle (5pt)  ++(4,0)
      circle (5pt);
      
      \path [fill,black] (-2.5,5.5)
      circle (5pt)   ++(4,0)
      circle (5pt)   ++(16,0)
      node[anchor=west] {\small$M_1=\Xi\,(1|3|1|2)=M(|2|0|1)+4$};
      \path [fill,blue] (-1.5,5.5)
      circle (5pt)  ++(4,0)
      circle (5pt)  ++(4,0)
      circle (5pt)  ++(4,0)
      circle (5pt);
      \path [fill,red] (-0.5,5.5)
      circle (5pt)  ++(4,0)
      circle (5pt);
      \path [fill,green] (0.5,5.5)
      circle (5pt)  ++(4,0)
      circle (5pt)  ++(4,0)
      circle (5pt);
      
      \path [fill,black] (-2.5,4.5)
      circle (5pt)   ++(4,0)
      circle (5pt)   ++(16,0)
      node[anchor=west] {\small$M_2=\Xi\,(1|4|1|2)=M(|3|0|1)+4$};
      \path [fill,blue] (-1.5,4.5)
      circle (5pt)  ++(4,0)
      circle (5pt)  ++(4,0)
      circle (5pt)  ++(4,0)
      circle (5pt)  ++(4,0)
      circle (5pt);
      \path [fill,red] (-0.5,4.5)
      circle (5pt)  ++(4,0)
      circle (5pt);
      \path [fill,green] (0.5,4.5)
      circle (5pt)  ++(4,0)
      circle (5pt)  ++(4,0)
      circle (5pt);

      \path [fill,black] (-2.5,3.5)
      circle (5pt)   ++(4,0)
      circle (5pt)   ++(16,0)
      node[anchor=west] {\small$M_3=\Xi\,(1|4|1|3)=M(|3|0|2)+4$};
      \path [fill,blue] (-1.5,3.5)
      circle (5pt)  ++(4,0)
      circle (5pt)  ++(4,0)
      circle (5pt)  ++(4,0)
      circle (5pt)  ++(4,0)
      circle (5pt);
      \path [fill,red] (-0.5,3.5)
      circle (5pt)  ++(4,0)
      circle (5pt);
      \path [fill,green] (0.5,3.5)
      circle (5pt)  ++(4,0)
      circle (5pt)  ++(4,0)
      circle (5pt)  ++(4,0)
      circle (5pt);

      \path [fill,black] (-2.5,2.5)
      circle (5pt)   ++(4,0)
      circle (5pt)   ++(16,0)
      node[anchor=west] {\small$M_4=\Xi\,(1|4|2|3)=M(|3|1|2)+4$};
      \path [fill,blue] (-1.5,2.5)
      circle (5pt)  ++(4,0)
      circle (5pt)  ++(4,0)
      circle (5pt)  ++(4,0)
      circle (5pt)  ++(4,0)
      circle (5pt);
      \path [fill,red] (-0.5,2.5)
      circle (5pt)  ++(4,0)
      circle (5pt)  ++(4,0)
      circle (5pt);
      \path [fill,green] (0.5,2.5)
      circle (5pt)  ++(4,0)
      circle (5pt)  ++(4,0)
      circle (5pt)  ++(4,0)
      circle (5pt);

      \draw  (-3,2) grid +(19 ,5);
      \draw[line width=2pt] (1,2) -- ++ (0,5);
      \draw[line width=2pt] (15,2) -- ++ (0,5);
      
      \foreach \x in {0,...,14} \draw (\x+1.5,1.5)  node {$\x$};
    \end{tikzpicture}
    \caption{A Maya $4$-cycle with shift $k=4$ for the choice
      $(|n_1|n_2|n_3)=(|3|1|2)$ and permutation $\bpi=(0,1,3,2)$. }
    \label{fig:44cyclic}
  \end{figure}
\end{example}

\subsection{Rational solutions of the $A_4$ Noumi-Yamada system}
In this section describe, the rational Hermite-type solutions to the
$A_4$-Painlev\'e system, and give examples in each signature class.

The $A_4$  Painlev\'e system consists of 5  equations in 5 unknowns $f_i=f_i(x),\; i=0,\ldots, 4$ and complex parameters $\alpha_i,\;i=0,\ldots, 5$:
\begin{equation}\label{eq:A4system}
  \begin{aligned}
    f_0' &= f_0(f_1-f_2+f_3-f_4) + \alpha_0\\
    f_1' &= f_1(f_2-f_3+f_4-f_0) + \alpha_1 \\
    f_2' &= f_2(f_3-f_4+f_0-f_1) + \alpha_2, \\
    f_3' &= f_3(f_4-f_0+f_1-f_2) + \alpha_3 \\
    f_4' &= f_4(f_0-f_1+f_2-f_3) + \alpha_4 
  \end{aligned}
\end{equation}
with normalization
\[ f_0+f_1+f_2+f_3+f_4=x.\]

Rational Hermite-type solutions of the $A_4$ system
\eqref{eq:A4system} correspond to chains of $5$-cyclic Maya diagrams
belonging to one of the following signature classes:
\[ (5), (3,1,1), (1,3,1), (1,1,3), (1,1,1,1,1).\]
% Using a cyclic permutation, no generality will be lost if restrict our attention to signatures
% \[ (5), (1,1,3), (1,1,1,1,1).\]
With the normalizations $\mu_0=0$, each $5$-cyclic Maya diagram may be
uniquely labelled by one of the above signatures, and a 4-tuple of
non-negative integers $(n_1, n_2, n_3, n_4)$.  For each of the above
signatures, the corresponding $k$-block coordinates are then given by
\begin{align*}
  &k=1 & &(5)& M(n_1,n_2,n_3,n_4)&:=\Xi(0,n_1,n_1+n_2, n_1+n_2+n_3,
                 n_1+n_2+n_3+n_4) \\  
  &k=3 & &(3,1,1) & M(n_1,n_2|n_3|n_4)&:=\Xi(0, n_1 , n_1+ n_2 | n_3 | n_4)\\
  &k=3 & &(1,3,1)  &M(|n_1,n_2,n_3|n_4)&:=\Xi(0| n_1, n_1+n_2,
                      n_1+n_2+n_3 | n_4)\\ 
  &k=3 & &(1,1,3)  &M(|n_1|n_2,n_3,n_4)&:=\Xi(0| n_1 | n_2, n_2+n_3,
                      n_2+n_3+n_4 )\\ 
  &k=5 & &(1,1,1,1,1) & M(|n_1|n_2|n_3|n_4)&:=\Xi(0| n_1 | n_2| n_3| n_4)
\end{align*}
Below, we exhibit examples belonging to each of these classes.

\begin{example}\label{ex:51}
  We exhibit a $(5,1)$-cyclic Maya diagram in the signature class
  $(5)$ by taking
  \[M_0=M(2,3,1,1)=\Xi(0,2,5,6,7),\] depicted in the first row of
  Figure~\ref{fig:51cyclic}.  The canonical flip sequence is
  $\bmu=(0,2,5,6,7)$.  By way of example, we choose the permutation
  $(34210)$, which gives the chain of Maya diagrams shown in Figure
  \ref{fig:51cyclic}.  Note that the permutation specifies the
  sequence of block coordinates that get shifted by one at each step
  of the cycle. This type of solutions with signature $(5)$ were
  already studied in \cite{filipuk2008symmetric}, and they are based
  on a genus 2 generalization of the generalized Hermite polynomials
  that appear in the solution of \PIV ($A_2$-Painlev\'e).
\begin{figure}[ht]
  \centering
\begin{tikzpicture}[scale=0.5]

  \path [fill] (0.5,6.5) ++(3,0)
  circle (5pt) ++(1,0)
  circle (5pt) ++ (1,0) circle (5pt)++ (2,0)
  circle (5pt) ++ (3,0)  node[anchor=west] {
    $M_0=\Xi(0,2,5,6,7)=M(2,3,1,1)$};

  \path [fill] (0.5,5.5) ++(3,0) circle (5pt) ++(1,0) circle (5pt) ++
  (1,0) circle (5pt) ++ (5,0)  node[anchor=west] {
    $M_1=\Xi(0,2,5,7,7)=M(2,3,2,0)$};

  \path [fill] (0.5,4.5) ++(3,0)
  circle (5pt) ++(1,0)
  circle (5pt) ++ (1,0) circle (5pt)++ (3,0)
  circle (5pt) ++ (2,0)  node[anchor=west] {
    $M_2=\Xi\,(0,2,5,7,8)=M(2,3,2,1)$};

  \path [fill] (0.5,3.5) ++(3,0)
  circle (5pt) ++(1,0)
  circle (5pt) ++ (1,0) circle (5pt) ++(1,0) circle (5pt) ++ (2,0)
  circle (5pt) ++ (2,0)  node[anchor=west] {
    $M_3=\Xi\,(0,2,6,7,8)=M(2,4,1,1)$};

  \path [fill] (0.5,2.5)  ++(4,0)
  circle (5pt) ++ (1,0) circle (5pt) ++(1,0) circle (5pt) ++ (2,0)
  circle (5pt) ++ (2,0)  node[anchor=west] {
    $M_4=\Xi\,(0,3,6,7,8)=M(3,3,1,1)$};

  \path [fill] (0.5,1.5)  ++(1,0) circle (5pt) ++(3,0)
  circle (5pt) ++ (1,0) circle (5pt) ++(1,0) circle (5pt) ++ (2,0)
  circle (5pt) ++ (2,0)  node[anchor=west] {
    $M_5=\Xi\,(1,3,6,7,8)=M(2,3,1,1)+1$};

  \path [fill] (0.5,1.5) ++(0,0) circle (5pt)
  ++(0,1) circle (5pt)
  ++(0,1) circle (5pt)
  ++(0,1) circle (5pt)
  ++(0,1) circle (5pt)
  ++(0,1) circle (5pt);

  \draw  (0,1) grid +(10 ,6);
  \draw[line width=2pt] (1,1) -- ++ (0,6);
  \draw[line width=2pt] (9,1) -- ++ (0,6);

  \foreach \x in {-1,...,8} \draw (\x+1.5,0.5)  node {$\x$};

\end{tikzpicture}
  
\caption{A Maya $5$-cycle with shift $k=1$ for the choice $(n_1,n_2,n_3,n_4)=(2,3,1,1)$ and permutation $\bpi=(34210)$. }
  \label{fig:51cyclic}
\end{figure}
\end{example}

We shall now provide the explicit construction of the rational
solution to the $A_4$-Painlev\'e system \eqref{eq:A4system}, by using
Proposition~\ref{prop:wtof} and Proposition~\ref{prop:Mtauw}. The
permutation $\bpi=(34210)$ on the canonical sequence
$\bmu=(0,2,5,6,7)$ produces the flip sequence
$\bmu_\bpi=(6,7,5,2,0)$. The pseudo-Wronskians corresponding to each
Maya diagram in the cycle are ordinary Wronskians, which will always
be the case with the normalization imposed in
Remark~\ref{rem:normalization}. They read (see
Figure~\ref{fig:51cyclic}):
\begin{align*}
  \tau_0&=\tau(2,3,4,6)\\
  \tau_1&=\tau(2,3,4)\\
  \tau_2&=\tau(2,3,4,7)\\
  \tau_3&=\tau(2,3,4,5,7)\\
  \tau_4&=\tau(3,4,5,7)\\
  \tau_5&=\tau(0,3,4,5,7)\propto \tau(2,3,4,6)
\end{align*}
The rational solution to the $5$-cyclic dressing
chain \eqref{eq:wachain} is given by  \eqref{eq:HM2w},  with
\[ (\sigma_0,\ldots, \sigma_4) = (1,-1,-1,1,-1),\quad (a_0,\ldots,
  a_4) = (-2,4,6,4,-14).\] Finally, the corresponding rational
solution to the $A_4$-Painlev\'e system \eqref{eq:Aeven} is given by
\[
 f_i(x) = - (\sigma_i+\sigma_{i+1})\frac{x}{2} +\frac{1}{\sqrt{2}}
  \frac{d}{dx} \log \frac{\tau_{i}(z)}{\tau_{i+2}(z)} ,\quad z=
  \frac{x}{\sqrt{2}},\quad \alpha_i = -\frac{a_i}2,\quad
  i=0,1,\ldots, 4 \mod 5
\]
\begin{example}
  We construct a degenerate example belonging to the $(5)$ signature
  class, by choosing $(n_1,n_2,n_3,n_4)=(1,1,2,0)$. The presence of
  $n_4=0$ means $M_0$ and $M_4$ have genus 1 instead of the generic
  genus 2.  The degeneracy occurs because the canonical flip sequence
  $\bmu=(0,1,2,4,4)$ contains two flips at the same site. Choosing the
  permutation $(42130)$, by way of example, produces the chain of Maya
  diagrams shown in Figure \ref{fig:51cyclicdegen}.  The explicit
  construction of the rational solutions follows the same steps as in
  the previous example, and we shall omit it here. It is worth noting,
  however, that due to the degenerate character of the chain, three
  linear combinations of $f_0,\dots,f_4$ will provide a solution to
  the lower rank $A_2$-Painlev\'e. If the two flips at the same site
  are performed consecutively in the cycle, the embedding of
  $A_2$ into $A_4$ is trivial and corresponds to setting
  two consecutive $f_i$ to zero. This is not the case in this example,
  as the flip sequence is $\bmu_\bpi=(4,2,1,4,0)$, which produces a
  non-trivial embedding.
% Note that permutation $(43251)$
%   cannot be applied because at the first step it would result in the
%   impossible symbol $[0,1,1,1,-1)$
\begin{figure}[ht]
  \centering
\begin{tikzpicture}[scale=0.5]

  \path [fill] (0.5,6.5) 
  ++(2,0)circle (5pt)
  ++(5,0)  node[anchor=west] {  $M_0=\Xi\,(0,1,2,4,4)= M(1,1,2,0)$};

  \path [fill] (0.5,5.5) 
  ++(2,0) circle (5pt) 
  ++(3,0) circle (5pt)
  ++ (2,0)  node[anchor=west] {$M_1=\Xi\,(0,1,2,4,5)= M(1,1,2,1)$};

  \path [fill] (0.5,4.5) 
  ++(2,0) circle (5pt)
  ++(1,0) circle (5pt)
  ++(2,0) circle (5pt)
  ++(2,0)  node[anchor=west] {  $M_2=\Xi\,(0,1,3,4,5)=M(1,2,1,1)$};

  \path [fill] (0.5,3.5) 
  ++(3,0) circle (5pt)
  ++(2,0) circle (5pt)
  ++(2,0)  node[anchor=west] {  $M_3=\Xi\,(0,2,3,4,5)= M(2,1,1,1)$};

  \path [fill] (0.5,2.5) 
  ++(3,0) circle (5pt)
  ++(4,0)  node[anchor=west] {  $M_4=\Xi\,(0,2,3,5,5)= M(2,1,2,0)$};

  \path [fill] (0.5,1.5) 
  ++(1,0) circle (5pt)
  ++(2,0) circle (5pt)
  ++(4,0)  node[anchor=west] {  $M_5=\Xi\,(1,2,3,5,5)=M(1,1,2,0)+1$};

  \path [fill] (0.5,1.5) ++(0,0) circle (5pt)
  ++(0,1) circle (5pt)
  ++(0,1) circle (5pt)
  ++(0,1) circle (5pt)
  ++(0,1) circle (5pt)
  ++(0,1) circle (5pt);

  \draw  (0,1) grid +(7 ,6);
  \draw[line width=2pt] (1,1) -- ++ (0,6);
  \draw[line width=2pt] (6,1) -- ++ (0,6);

  \foreach \x in {-1,...,5} \draw (\x+1.5,0.5)  node {$\x$};

\end{tikzpicture}
  
  \caption{A degenerate  Maya $5$-cycle with $k=1$ for the choice  $(n_1,n_2,n_3,n_4)=(1,1,2,0)$ and permutation $\bpi=(42130)$.}
  \label{fig:51cyclicdegen}
\end{figure}
\end{example}

\begin{example}
We construct a $(5,3)$-cyclic Maya diagram in the signature class $(1,1,3)$ by choosing $(n_1,n_2,n_3,n_4)=(3,1,1,2)$, which means that the first Maya diagram has $3$-block coordinates $(0|3|1,2,4)$. The canonical flip sequence is given by $\bmu=\Theta\,(0|3|1,2,4)=(0, {\color{red}10},{\color{blue} 5,8,14})$.
The permutation $(41230)$ gives the chain of Maya diagrams shown in Figure
  \ref{fig:53cyclic}. Note that, as in Example~\ref{ex:51}, the permutation specifies the order in which the $3$-block coordinates are shifted by +1 in the subsequent steps of the cycle. This type of solutions in the signature class $(1,1,3)$ were not given in \cite{filipuk2008symmetric}, and they are new to the best of our knowledge.

\begin{figure}[ht]
  \centering
\begin{tikzpicture}[scale=0.5]
  \path [fill,red] (0.5,6.5)  
  ++(2,0) circle (5pt)
  ++(3,0) circle (5pt)
  ++(3,0) circle (5pt) ;

  \path [fill,blue] (0.5,6.5)  
  ++(3,0) circle (5pt)
  ++(6,0) circle (5pt) 
  ++(3,0) circle (5pt) ;

  \path (17.5,6.5)  node[anchor=west] {$M_0=\Xi\,(0|3|1,2,4)=M(|3|1,1,2)$};

  \path [fill,red] (0.5,5.5)  
  ++(2,0) circle (5pt)
  ++(3,0) circle (5pt)
  ++(3,0) circle (5pt) ;

  \path [fill,blue] (0.5,5.5)  
  ++(3,0) circle (5pt)
  ++(6,0) circle (5pt) 
  ++(3,0) circle (5pt) 
  ++(3,0) circle (5pt) ;
  \path (17.5,5.5)  node[anchor=west] {$M_1=\Xi\,(0|3|1,2,5)=M(|3|1,1,3)$};

  \path [fill,red] (0.5,4.5)  
  ++(2,0) circle (5pt)
  ++(3,0) circle (5pt)
  ++(3,0) circle (5pt)
  ++(3,0) circle (5pt) ;

  \path [fill,blue] (0.5,4.5)  
  ++(3,0) circle (5pt)
  ++(6,0) circle (5pt) 
  ++(3,0) circle (5pt) 
  ++(3,0) circle (5pt) ;
  \path (17.5,4.5)  node[anchor=west] {$M_2=\Xi\,(0|4|1,2,5)=M(|4|1,1,3)$};

  \path [fill,red] (0.5,3.5)  
  ++(2,0) circle (5pt)
  ++(3,0) circle (5pt)
  ++(3,0) circle (5pt)
  ++(3,0) circle (5pt) ;

  \path [fill,blue] (0.5,3.5)  
  ++(3,0) circle (5pt)
  ++(3,0) circle (5pt)
  ++(3,0) circle (5pt) 
  ++(3,0) circle (5pt) 
  ++(3,0) circle (5pt) ;
  \path (17.5,3.5)  node[anchor=west] {$M_3=\Xi\,(0|4|2,2,5)=M(|4|2,0,3)$};

  \path [fill,red] (0.5,2.5)  
  ++(2,0) circle (5pt)
  ++(3,0) circle (5pt)
  ++(3,0) circle (5pt)
  ++(3,0) circle (5pt) ;
  \path [fill,blue] (0.5,2.5)  
  ++(3,0) circle (5pt)
  ++(3,0) circle (5pt)
  ++(6,0) circle (5pt) 
  ++(3,0) circle (5pt) ;
  \path (17.5,2.5)  node[anchor=west] {$M_4=\Xi\,(0|4|2,3,5)=M(|4|2,1,2)$};

  \path [fill,black] (0.5,1.5)  
  ++(1,0) circle (5pt);  
  \path [fill,red] (0.5,1.5)  
  ++(2,0) circle (5pt)
  ++(3,0) circle (5pt)
  ++(3,0) circle (5pt)
  ++(3,0) circle (5pt) ;
  \path [fill,blue] (0.5,1.5)  
  ++(3,0) circle (5pt)
  ++(3,0) circle (5pt)
  ++(6,0) circle (5pt) 
  ++(3,0) circle (5pt) ;
  \path (17.5,1.5)  node[anchor=west] {$M_5=\Xi\,(1|4|2,3,5)=M(|3|1,1,2)+3$};

  \path [fill,blue] (0.5,1.5) ++(0,0) circle (5pt)
  ++(0,1) circle (5pt)
  ++(0,1) circle (5pt)
  ++(0,1) circle (5pt)
  ++(0,1) circle (5pt)
  ++(0,1) circle (5pt);

  \path [fill,red] (-.5,1.5) ++(0,0) circle (5pt)
  ++(0,1) circle (5pt)
  ++(0,1) circle (5pt)
  ++(0,1) circle (5pt)
  ++(0,1) circle (5pt)
  ++(0,1) circle (5pt);

  \path [fill,black] (-1.5,1.5) ++(0,0) circle (5pt)
  ++(0,1) circle (5pt)
  ++(0,1) circle (5pt)
  ++(0,1) circle (5pt)
  ++(0,1) circle (5pt)
  ++(0,1) circle (5pt);

  \draw  (-2,1) grid +(19 ,6);
  \draw[line width=2pt] (1,1) -- ++ (0,6);
  \draw[line width=2pt] (16,1) -- ++ (0,6);

  \foreach \x in {0,...,15} \draw (\x+1.5,0.5)  node {$\x$};
\end{tikzpicture}
\caption{A Maya $5$-cycle with shift $k=3$ for the choice
  $(|n_1|n_2,n_3,n_4)=(|3|1,1,2)$ and permutation $\bpi=(41230)$. }
  \label{fig:53cyclic}
\end{figure}

We proceed to build the explicit rational solution to the
$A_4$-Painlev\'e system \eqref{eq:A4system}. In this case, the
permutation $\bpi=(41230)$ on the canonical sequence
$\bmu=(0,10,5,8,14)$ produces the flip sequence
$\bmu_\bpi=(14,10,5,8,0)$, so that the values of the $a_i$
parameters given by \eqref{eq:HM2w} become
$(a_0,a_1,a_2,a_3,a_4)=(8,10,-6,16,-34)$. The
pseudo-Wronskians corresponding to each Maya diagram in the cycle are
ordinary Wronskians, which will always be the case with the
normalization imposed in Remark~\ref{rem:normalization}. They read
(see Figure~\ref{fig:53cyclic}):
\begin{align*}
\tau_0&=\tau(1,2,4,7,8,11)\\
\tau_1&=\tau(1,2,4,7,8,{11},{14})\\
\tau_2&=\tau(1,2,4,7,8,{10},{11},{14})\\
\tau_3&=\tau(1,2,4,5,7,8,{10},{11},{14})\\
\tau_4&=\tau(1,2,4,5,7,{10},{11},{14})
\end{align*}
The rational solution to the $5$-cyclic dressing
chain \eqref{eq:wachain} is given by  \eqref{eq:HM2w},  with
\[ (\sigma_0,\ldots, \sigma_4) = (-1,1,-1,-1,-1),\quad (a_0,\ldots,
  a_4) = (-6,-12,8,20,-16).\] The corresponding rational solution to
the $A_4$-Painlev\'e system \eqref{eq:Aeven} is given by
\[
  f_i(x) = - (\sigma_i+\sigma_{i+1})\frac{x}{6} +\frac{1}{\sqrt{6}}
  \frac{d}{dx} \log \frac{\tau_{i}(z)}{\tau_{i+2}(z)} ,\quad z=
  \frac{x}{\sqrt{6}},\quad \alpha_i = -\frac{a_i}6,\quad
  i=0,1,\ldots, 4 \mod 5
\]

\end{example}

\begin{example}
  We construct a $(5,5)$-cyclic Maya diagram in the signature class
  $(1,1,1,1,1)$ by choosing $(n_1n_2n_3n_4)=(2,3,0,1)$, which means
  that the first Maya diagram has $5$-block coordinates
  $(0|2|3|0|1)$. The canonical flip sequence is given by
  \[\bmu=\Theta\,(0|2|3|0|1)=(0,{\color{red}11},{\color{blue}17},{\color{brown}3},{\color{green}
      9}).\] The permutation $(32410)$ gives the chain of Maya
  diagrams shown in Figure \ref{fig:55cyclic}.  Note that, as it
  happens in the previous examples, the permutation specifies the
  order in which the $5$-block coordinates are shifted by +1 in the
  subsequent steps of the cycle. This type of solutions with signature
  $(1,1,1,1,1)$ were already studied in \cite{filipuk2008symmetric},
  and they are based on a generalization of the Okamoto polynomials
  that appear in the solution of \PIV ($A_2$-Painlev\'e).
    
\begin{figure}[ht]
  \centering
\begin{tikzpicture}[scale=0.4]
  \path [fill,red] (0.5,6.5)  
  ++(2,0) circle (5pt)
  ++(5,0) circle (5pt);
  \path [fill,blue] (0.5,6.5)  
  ++(3,0) circle (5pt)
  ++(5,0) circle (5pt)
  ++(5,0) circle (5pt) ;
  \path [fill,green] (0.5,6.5)  
  ++(5,0) circle (5pt);
  \path (20.5,6.5)  node[anchor=west] {$\small M_0=\Xi\,(0|2|3|0|1)=M(|2|3|0|1)$};

  \path [fill,red] (0.5,5.5)  
  ++(2,0) circle (5pt)
  ++(5,0) circle (5pt);
  \path [fill,blue] (0.5,5.5)  
  ++(3,0) circle (5pt)
  ++(5,0) circle (5pt)
  ++(5,0) circle (5pt) ;
  \path [fill,brown] (0.5,5.5)  
  ++(4,0) circle (5pt);
  \path [fill,green] (0.5,5.5)  
  ++(5,0) circle (5pt);
  \path (20.5,5.5)  node[anchor=west] {$\small M_1=\Xi\,(0|2|3|1|1)$};

  \path [fill,red] (0.5,4.5)  
  ++(2,0) circle (5pt)
  ++(5,0) circle (5pt);
  \path [fill,blue] (0.5,4.5)  
  ++(3,0) circle (5pt)
  ++(5,0) circle (5pt)
  ++(5,0) circle (5pt) 
  ++(5,0) circle (5pt) ;
  \path [fill,brown] (0.5,4.5)  
  ++(4,0) circle (5pt);
  \path [fill,green] (0.5,4.5)  
  ++(5,0) circle (5pt);
  \path (20.5,4.5)  node[anchor=west] {$\small M_2=\Xi\,(0|2|4|1|1)$};

  \path [fill,red] (0.5,3.5)  
  ++(2,0) circle (5pt)
  ++(5,0) circle (5pt);
  \path [fill,blue] (0.5,3.5)  
  ++(3,0) circle (5pt)
  ++(5,0) circle (5pt)
  ++(5,0) circle (5pt) 
  ++(5,0) circle (5pt) ;
  \path [fill,brown] (0.5,3.5)  
  ++(4,0) circle (5pt);
  \path [fill,green] (0.5,3.5)  
  ++(5,0) circle (5pt)
  ++(5,0) circle (5pt);
  \path (20.5,3.5)  node[anchor=west] {$\small M_3=\Xi\,(0|2|4|1|2)$};

  \path [fill,red] (0.5,2.5)  
  ++(2,0) circle (5pt)
  ++(5,0) circle (5pt)
  ++(5,0) circle (5pt);
  \path [fill,blue] (0.5,2.5)  
  ++(3,0) circle (5pt)
  ++(5,0) circle (5pt)
  ++(5,0) circle (5pt) 
  ++(5,0) circle (5pt) ;
  \path [fill,brown] (0.5,2.5)  
  ++(4,0) circle (5pt);
  \path [fill,green] (0.5,2.5)  
  ++(5,0) circle (5pt)
  ++(5,0) circle (5pt);
  \path (20.5,2.5)  node[anchor=west] {$\small M_4=\Xi\,(0|3|4|1|2)$};

  \path [fill,black] (0.5,1.5)  
  ++(1,0) circle (5pt);  
  \path [fill,red] (0.5,1.5)  
  ++(2,0) circle (5pt)
  ++(5,0) circle (5pt)
  ++(5,0) circle (5pt);
  \path [fill,blue] (0.5,1.5)  
  ++(3,0) circle (5pt)
  ++(5,0) circle (5pt)
  ++(5,0) circle (5pt) 
  ++(5,0) circle (5pt) ;
  \path [fill,brown] (0.5,1.5)  
  ++(4,0) circle (5pt);
  \path [fill,green] (0.5,1.5)  
  ++(5,0) circle (5pt)
  ++(5,0) circle (5pt);
  \path (20.5,1.5)  node[anchor=west] {$\small M_5=\Xi(1|3|4|1|2)=M(|2|3|0|1)+5$};

  \path [fill,green] (0.5,1.5) ++(0,0) circle (5pt)
  ++(0,1) circle (5pt)
  ++(0,1) circle (5pt)
  ++(0,1) circle (5pt)
  ++(0,1) circle (5pt)
  ++(0,1) circle (5pt);
  \path [fill,brown] (-0.5,1.5) ++(0,0) circle (5pt)
  ++(0,1) circle (5pt)
  ++(0,1) circle (5pt)
  ++(0,1) circle (5pt)
  ++(0,1) circle (5pt)
  ++(0,1) circle (5pt);
  \path [fill,blue] (-1.5,1.5) ++(0,0) circle (5pt)
  ++(0,1) circle (5pt)
  ++(0,1) circle (5pt)
  ++(0,1) circle (5pt)
  ++(0,1) circle (5pt)
  ++(0,1) circle (5pt);

  \path [fill,red] (-2.5,1.5) ++(0,0) circle (5pt)
  ++(0,1) circle (5pt)
  ++(0,1) circle (5pt)
  ++(0,1) circle (5pt)
  ++(0,1) circle (5pt)
  ++(0,1) circle (5pt);

  \path [fill,black] (-3.5,1.5) ++(0,0) circle (5pt)
  ++(0,1) circle (5pt)
  ++(0,1) circle (5pt)
  ++(0,1) circle (5pt)
  ++(0,1) circle (5pt)
  ++(0,1) circle (5pt);

  \draw  (-4,1) grid +(24 ,6);
  \draw[line width=2pt] (1,1) -- ++ (0,6);
  \draw[line width=2pt] (19,1) -- ++ (0,6);

  \foreach \x in {0,...,18} \draw (\x+1.5,0.5)  node {$\x$};
\end{tikzpicture}
  \caption{A Maya $5$-cycle with shift $k=5$ for the choice $(n_1,n_2,n_3,n_4)=(2,3,0,1)$ and permutation $\bpi=(32410)$.}
  \label{fig:55cyclic}
\end{figure}

We proceed to build the explicit rational solution to the
$A_4$-Painlev\'e system \eqref{eq:A4system}. In this case, the flip
sequence is given by $\bmu_\bpi=(3,17,9,11,0)$. The corresponding
Hermite Wronskians are shown below (see Figure \ref{fig:55cyclic}):
\begin{align*}
\tau&=\tau(1,2,4,6,7,{12})\\
\tau&=\tau(1,2,3,4,6,7,{12})\\
\tau&=\tau(1,2,3,4,6,7,{12},{17})\\
\tau&=\tau(1,2,3,4,6,7,9,{12},{17})\\
\tau&=\tau(1,2,3,4,6,7,9,{11},{12},{17})
\end{align*}
The rational solution to the $5$-cyclic dressing chain
\eqref{eq:wachain} is given by \eqref{eq:HM2w}, with
\[ (\sigma_0,\ldots, \sigma_4) = (-1,-1,-1,-1,-1),\quad (a_0,\ldots,
  a_4) =(-28,16,-4,22,-16).\] The corresponding rational solution to
the $A_4$-Painlev\'e system \eqref{eq:Aeven} is given by
\[
  f_i(x) =\frac{x}{5} +\frac{1}{\sqrt{10}}
  \frac{d}{dx} \log \frac{\tau_{i}(z)}{\tau_{i+2}(z)} ,\quad z=
  \frac{x}{\sqrt{10}},\quad \alpha_i = -\frac{a_i}{10},\quad
  i=0,1,\ldots, 4 \mod 5
\]

\end{example}

\section*{Acknowledgements}
The research of DGU has been supported in part by Spanish MINECO-FEDER
Grant MTM2015-65888-C4-3 and by the ICMAT-Severo Ochoa project
SEV-2015-0554. The research of RM was supported in part by NSERC grant
RGPIN-228057-2009. 

\strut\hfill

% \noindent
%     {\bf A Remark on the references:}
%     The references should be ordered alphabetically according to the
% authors in the list and the authors' initials should be after their surname without a
% comma, e.g. Ovsienko V and Roger C. Please see below.

\providecommand{\bysame}{\leavevmode\hbox to3em{\hrulefill}\thinspace}
\providecommand{\MR}{\relax\ifhmode\unskip\space\fi MR }
\providecommand{\MRhref}[2]{%
  \href{http://www.ams.org/mathscinet-getitem?mr=#1}{#2}
}
\providecommand{\href}[2]{#2}

\end{document}